\newcommand{\be}{\begin{equation}}
\newcommand{\bea}{\begin{eqnarray}}
\newcommand{\ee}{\end{equation}}
\newcommand{\eea}{\end{eqnarray}}
\newcommand{\nn}{\nonumber}

\newcommand{\qa}{\alpha}
\newcommand{\qb}{\beta}

\newcommand{\qy}{\theta}

\newcommand{\ql}{\lambda}

\newcommand{\qr}{\rho}
\newcommand{\qs}{\sigma}

\newcommand{\qf}{\varphi}
\newcommand{\qF}{\Phi}

\newcommand{\dagg}{^{\dag}}

\renewcommand{\Re}{{\rm Re}\,}
\renewcommand{\Im}{{\rm Im}\,}
\newcommand{\rd}{{\rm d}}

\newcommand{\fr}[2]{{\textstyle \frac{#1}{#2}}}

\newcommand{\CC}{{\mathbb C}}

\newcommand{\bits}{ \{0,1\} }

\newcommand{\cF}{{\mathcal F}}

\newcommand{\vecx}{{\boldsymbol{x}}}

\newcommand{\isdef}{\stackrel{\rm def}{=}}

\documentclass[10pt,twocolumn]{article}

\usepackage{amsmath,amsthm,amssymb}
\usepackage{amsfonts}
\usepackage{graphicx}
\usepackage{url}
\usepackage{a4wide}

\newtheorem{theorem}{Theorem}[section]

\newtheorem{lemma}[theorem]{Lemma}
\begin{document}

\setlength{\parindent}{0mm}
\setlength{\hoffset}{0mm}
\setlength{\oddsidemargin}{-11mm}
\setlength{\evensidemargin}{-11mm}
\setlength{\textwidth}{190mm}
\setlength{\columnwidth}{95mm}
\setlength{\columnsep}{5mm}
\setlength{\columnseprule}{0mm}
\setlength{\voffset}{-10mm}
\setlength{\topmargin}{0mm}
\setlength{\headheight}{0mm}
\setlength{\headsep}{0mm}
\setlength{\textheight}{240mm}

\title{Minutia-pair spectral representations for fingerprint template protection}

\author{Taras Stanko and Boris \v{S}kori\'{c}
\\
{\small Eindhoven University of Technology, The Netherlands}}

\date{ }



\maketitle

\begin{abstract}
\noindent
We introduce a new fixed-length representation of fingerprint minutiae,
for use in template protection.
It is similar to the `spectral minutiae' representation of Xu et al.
but is based on coordinate differences between pairs of minutiae.
Our technique has the advantage that it does not discard the phase information
of the spectral functions.
We show that the fingerprint matching performance (Equal Error Rate) is comparable
to that of the original spectral minutiae representation, while the speed is improved.
\end{abstract}



\section{Introduction}
\label{sec:intro}

\subsection{Privacy-preserving storage of biometric data}

Biometrics-based authentication has become popular because 
of its great convenience. Biometrics cannot be forgotten or accidentally left at home.
While biometric data is not strictly speaking secret
(we are after all leaving a trail of fingerprints, DNA etc. behind us),
it is important to protect biometric data for various reasons,
the most important of which is privacy.
Unprotected storage of biometric data would reveal medical conditions and would allow for
cross-matching entries in different databases.
Furthermore, large-scale availability of biometric data would make it easier for 
malevolent parties to 
leave misleading traces at at crime scene.
(E.g. artificial fingerprints \cite{matsumoto2002}, synthesized DNA \cite{FWDG2010}.)

One of the easiest ways to properly protect a biometric database against 
breaches and insider attacks is to store biometrics in {\em hashed} form,
just like passwords, but with the addition of an error-correction step
to get rid of the measurement noise.
To prevent critical leakage from the error correction redundancy data, one 
uses a {\em Helper Data System} (HDS) \cite{LT2003,dGSdVL,SAS2016},
for instance a {\em Fuzzy Extractor} or a {\em Secure Sketch} \cite{JW99,DORS2008,CFPRS2016}.

A HDS typically makes use of an error-correcting code and hence needs
a {\em fixed-length representation} of the biometric. 
Such a representation is not straightforward when 
the measurement noise can cause features of the biometric to appear or disappear,
due to e.g. occlusion of iris areas or fuzziness of fingerprint minutiae.
A very useful fixed-length representation called {\em spectral minutiae}
was introduced by Xu et al. \cite{XVBKAG2009,XV2009,XV2009CISP,XuVeldhuis2010}.
A Fourier-like spectral function is built up on a fixed discrete grid, in such a way that each detected 
fingerprint minutia adds a contribution to the function. 
Comparison of spectral functions is robust against changes in the number of available
biometric features.

\subsection{Contributions and outline}
\label{sec:contrib}

We have the following results regarding spectral representations of fingerprint minutiae.
\begin{itemize}
\item
We introduce spectral functions based on pairs of minutiae.
By working with coordinate {\em differences} we immediately
obtain a translation-invariant representation.
Whereas Xu et al.'s spectral functions have to discard phase information
in order to achieve translation invariance, our method retains phase information.
\item
We test our pair-based spectral minutiae matching technique
on two fingerprint databases.
The achieved Equal Error Rate is comparable to Xu et al.
\item
Our fingerprint matching is faster even though we have to sum over minutia pairs instead of
individual minutiae.
The speedup is due to the fact that we need fewer grid points on which to compute the spectral function.
\item
A further speedup can be obtained by skipping one laborious step in the verification procedure:
rotating the fingerprint so as to obtain optimal alignment with the enrolled fingerprint. 
Skipping this step leads only to a minimal penalty in terms of False Acceptance Rate and False Rejection Rate.
\end{itemize}

In Section~\ref{sec:prelim} we briefly review Helper Data Systems and spectral minutiae 
functions.
In Section~\ref{sec:motivation} we discuss the drawbacks of Xu et al.'s 
spectral minutiae technique.
We introduce our minutia pair approach in Section~\ref{sec:pair},
and we study its fingerprint matching performance in Section~\ref{sec:results}.
Section~\ref{sec:efficiency} discusses the computational efficiency of the
verification procedure.

\section{Preliminaries} 
\label{sec:prelim}

\subsection{Notation and terminology}
\label{sec:notation}

We denote the number of minutiae found in a fingerprint by~$Z$.
The coordinates of the $j$'th minutia are $x_j,y_j$ and its orientation
is $\qy_j$. Let $f$ be a function of two real-valued arguments.
The two-dimensional Fourier transform $\tilde f=\cF f$ is defined as
$\tilde f(k_x,k_y)=\int_{-\infty}^\infty\! f(x,y)e^{-ik_x x-ik_y y}\rd x \rd y$.
The inverse relation $f=\cF^{-1}\tilde f$ is given by
$f(x,y)=(\fr1{2\pi})^2\int_{-\infty}^\infty\! \tilde f(k_x,k_y)e^{ik_x x+ik_y y}\rd k_x \rd k_y$.

The complex conjugate of $z\in\CC$ is written as $z^*$.
The hermitean conjugate $M\dagg$ of a matrix $M$ is given by
$(M\dagg)_{ij}=M_{ji}^*$.
The inner product of two complex vectors $u,v$ is $\langle u,v\rangle=u\dagg v$.
The Pearson correlation coefficient of two length-$n$ vectors
is defined as
$\qr(u,v)=\fr1{n}\langle \frac{u-u_{\rm av}}{\qs_u},\frac{v-v_{\rm av}}{\qs_v}\rangle$,
where $u_{\rm av}=\fr1n\sum_i u_i$ and $\qs^2_u=\fr1n\sum_i |u_i-u_{\rm av}|^2$.

We will use the abbreviations FR = False Reject, FRR = False Reject Rate, FA = False Accept, 
FAR = False Accept Rate, EER = Equal Error Rate, ROC = Receiver Operating Characteristic.

\subsection{Helper Data Systems}
\label{sec:HDS}

A Helper Data System (HDS) for a (possibly non-discrete) source 
consists of two functions, {\tt Gen} and {\tt Rec}.
Given an enrollment measurement $X$ of the source, {\tt Gen} produces 
redundancy data $W\in\bits^*$ called {\em helper data} and a secret string~$S$.
The helper data is stored. 
The storage is considered insecure, i.e. attackers learn~$W$.
At some later time, a verification measurement is performed, yielding outcome $X'\approx X$
which is a noisy version of~$X$.
The {\tt Rec} function takes as input $X'$ and $W$. 
It outputs an estimator $\hat S$ which should equal $S$ if the noise was not excessive.
In a general HDS, there is no constraint on the distribution of~$S$.
A desirable property is that $S$ has high entropy given~$W$.

A HDS is the perfect primitive for privacy protection of biometric databases
against inside attackers and intruders, 
who typically obtain access not only to stored data but also to decryption keys.
The HDS creates a noiseless secret and thus makes it possible to protect
biometric secrets in the same way as passwords: by hashing.
For every enrolled user, the database contains $W$ and a hash $\chi(S)$. 
In the verification phase, the hash of the reconstructed $\hat S$ is compared against the stored $\chi(S)$.
Ideally, $W$ contains just enough information to allow for the error correction, and does 
not leak any privacy-sensitive information about the raw biometric~$X$.
Furthermore, if $\chi$ is a properly chosen one-way function and $S$ has enough entropy given $W$,
the hash value $\chi(S)$ does not reveal~$S$.

HDSs for discrete sources \cite{BBCS1991,JW99,Boy04,DORS2008,CFPRS2016}
and continuum sources \cite{LT2003,VTOSS10,dGSdVL,SAS2016}
are a well studied topic. 
Typically a HDS uses an error correcting code, which requires that the biometric measurement 
is turned into a discrete fixed-length representation.

\subsection{Spectral representation of minutiae}
\label{sec:prelimspectral}

Subsequent measurements of the same finger may not always result in
the same set of observed minutiae. 
This is problematic if one needs a fixed-length representation of a fingerprint,
e.g.\;when a HDS is used.
The technique of {\em spectral minutiae} was introduced by Xu et al.
\cite{XVBKAG2009,XV2009,XV2009CISP}
as a way to obtain a fixed-length representation.
The set of enrolled minutiae is turned into a function $f_\qs(x,y)$ on the $xy$-plane
by summing narrow Gaussian peaks (with width $\qs$) centered on the minutia locations;
then a translation-invariant expression $g_\qs$ is obtained by taking the absolute value of
the Fourier transform,
\be
	g_\qs(k_x,k_y)=|\tilde f_\qs(k_x,k_y)|=e^{-\fr{\qs^2}2(k_x^2+k_y^2)}
	\left|\sum_{j=1}^Z e^{-ik_x x_j -i k_y y_j}\right|.
\label{spectralXu}
\ee
In order to get an expression with simple behaviour under rotation and scaling,
they sampled $g_\qs$ on a log-polar grid.
Let $k_x(\qa,\qb)=e^\qa\cos\qb$ and $k_y(\qa,\qb)=e^\qa\sin\qb$
where $\qa,\qb$ are sampled with equal spacing.
A matrix $G^\qs$ is constructed as $G^\qs_{\qa\qb}=g_\qs(k_x(\qa,\qb),k_y(\qa,\qb))$.
Under the combination of scaling and rotation, 
${x_j \choose y_j}\mapsto
\left(\begin{matrix} \cos\qf & \sin\qf \cr -\sin\qf&\cos\qf \end{matrix}\right)
{\ql x_j\choose\ql y_j}$ 
for all~$j$,
the $G^\qs$ transforms as $G^\qs_{\qa\qb}\mapsto G^{\qs/\ql}_{\qa+\ln\ql,\qb+\qf}$.
For small $\qs$ it holds that $\qs/\ql\approx\qs$ and hence the transform is almost equal
to a shift on the $\qa\qb$-grid.\footnote{
The effect on $\qs$ was not explicitly mentioned in the work of Xu et al.
}
Xu et al. investigated fingerprint matching in the spectral minutiae domain
by looking at the Pearson correlation between a freshly obtained $G^\qs$ and the enrolled~$G^\qs$.
Their procedure included a search to find values $\ql,\qf$ that maximise the correlation.
It turned out that in practice one can fix $\ql=1$ and that the $\qf$-search can be restricted to 
the interval from $-10^\circ$ to $+10^\circ$, in steps of~$2^\circ$. 

In order to extract more information from a fingerprint Xu et al introduced a variant
of the $g_\qs$ function which contains information about the minutia orientations~$\qy_j$.
They inserted a factor $(k_x\cos\qy_j+k_y\sin\qy_j)$ or $e^{i\qy_j}$ into the summation 
in $g_\qs$ (\ref{spectralXu}). 
Unsurprisingly,
using information from both the ordinary $G^\qs$ representation and the orientation-containing variant
yielded better results (in terms of e.g. ROC curves and EER) than using only a single 
representation.

Xu et al also investigated a minutiae representation that is fully invariant under
translation, rotation and scaling. 
Let $H^\qs=\cF G^\qs$ be the discrete Fourier transform of $G^\qs_{\qa\qb}$
with respect to $\qa$ and $\qb$; then scalings and rotations have the effect of
merely producing a phase factor multiplying $H^\qs$; 
the absolute value $|H^\qs|$ is fully invariant.
However, it turned out that fingerprint matching in the $|H^\qs|$-domain does not perform well.

\section{Motivation}
\label{sec:motivation}

The spectral minutiae technique as developed by Xu et al
\cite{XVBKAG2009,XV2009,XV2009CISP} has a number of unsatisfactory aspects.
\begin{enumerate}
\item
Translation invariance is obtained by taking the absolute value of a Fourier transform.
This step discards a lot of information.
\item
Xu et al conclude that the scaling factor $\ql$ does not have to be taken into account,
since it is always close to~1.
But in their best fingerprint matching implementation
they still apply logarithmic sampling in the radial $k$-direction,
$\sqrt{k_x^2+k_y^2}=e^\qa$.
Such sampling does not match the radial information density in the fingerprint and
hence makes it necessary to take many many samples than in the case of linear sampling.
\item
In combination with a HDS, the $\qf$-search is time consuming. 
This is caused not by the repeated re-computation of the score, but
by the fact that in a full HDS every $\qf$-attempt needs
an evaluation of the {\tt Rec} function and the computation of a hash.
\end{enumerate}

We address the first issue by introducing a spectral representation 
that is based on coordinate differences $\vecx_a-\vecx_b$ only.
The advantage is immediate translation invariance without information loss,
enabling us to work with fewer samples.
The drawback is that each summation over $Z$ minutiae is replaced by 
a summation over ${Z\choose 2}$ pairs.
The overall effect on the computation time during reconstruction is a 
tradeoff between these two.  
In Section~\ref{sec:efficiency} we show that
the tradeoff works in our advantage.

We address the second issue by performing a Fourier transform 
{\em only in the angular direction}.
In the radial direction our sampling occurs in the spatial domain and is linear. 

The third issue could be addressed by developing a method to quickly 
determine the global orientation of a captured fingerprint image.
(Knowledge of the global orientation, even if inaccurate, reduces the search space.
Furthermore,
storing the global orientation during enrolment as helper data
does not leak sensitive information.)
However, with our pair-based spectral representation it turns out that
executing the $\qf$-search yields only a very modest performance improvement;
the search may as well be omitted.
In Section~\ref{sec:imagerotation} we show the difference in performance.

\section{The minutia-pair approach}
\label{sec:pair}

\subsection{Definitions and properties}
\label{sec:defsprops}

\begin{figure}[t]
\vskip-5mm
\centering
\includegraphics[width=60mm,height=45mm]{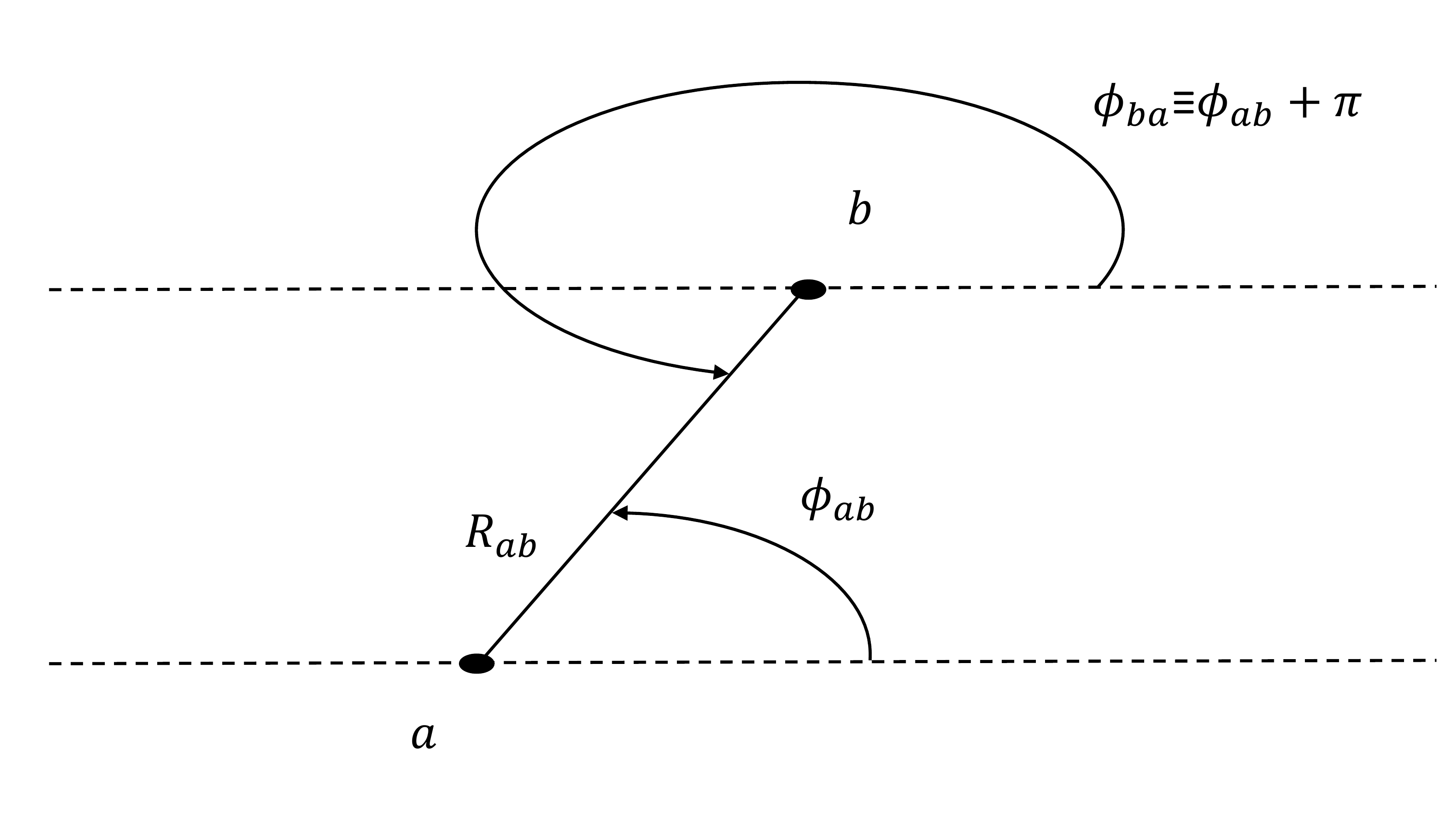}	
\caption{\it Distance $R_{ab}$ and angle $\qf_{ab}$ for a minutia pair. } 
\label{fig:AngleDistancePair}
\end{figure}

Let $R_{ab}=|\vecx_a-\vecx_b|$
and let $\tan\qf_{ab}=(y_a-y_b)/(x_a-x_b)$ for minutiae $a,b\in\{1,\ldots,Z\}$.
See Fig.\,\ref{fig:AngleDistancePair}.
We define two translation-invariant spectral functions as follows
\bea
	L_\vecx(q,w) &\isdef & \sum_{ {a,b\in\{1,\ldots,Z\}} \atop{a\neq b}} e^{iq\qf_{ab}}e^{iw\ln R_{ab}}
\label{defL}
	\\
	L_{\vecx\qy}(q,w) &\isdef &
	\sum_{ {a,b\in\{1,\ldots,Z\}} \atop{a\neq b}} e^{iq\qf_{ab}} e^{iw\ln R_{ab}} e^{i(\qy_a-\qy_b)}.
\label{defLtheta}
\eea
Here the subscript $\vecx$ denotes the set of minutia locations, and likewise $\qy$ stands for the set of
minutia orientations.
We call the functions $L_\vecx,L_{\vecx\qy}$ `spectral' because (\ref{defL}) is the Fourier transform 
(with respect to the radial coordinate $\ln R$ and the angle $\qf$)
of a sum of delta functions centered on the values $\vecx_a-\vecx_b$ in the plane.

Let $\qF=\left(\begin{matrix} \cos\qf & -\sin\qf \cr \sin\qf&\cos\qf \end{matrix}\right)$ be a rotation matrix.
Our spectral functions (\ref{defL}),(\ref{defLtheta}) have simple behaviour under the
combined scaling and rotation $\vecx_j\mapsto \ql\qF\vecx_j$, $\qy_j\mapsto\qy_j+\qf$,
\bea
	L_{\ql\qF\vecx}(q,w) &=& e^{iq\qf}e^{iw\ln\ql}L_\vecx(q,w)
\label{phaseLx}
	\\
	L_{\ql\qF\vecx,\qy+\qf}(q,w) &=& e^{iq\qf}e^{iw\ln\ql}L_{\vecx\qy}(q,w).
\label{phaseLxtheta}
\eea
Note that the absolute values $|L_\vecx(q,w)|$, $|L_{\vecx\qy}(q,w)|$ are invariant under translation,
scaling and rotation.
Without giving details we mention that, unfortunately, fingerprint matching based on 
$|L_\vecx|$, $|L_{\vecx\qy}|$ without the phase information performs badly.

Similar to Xu et al we need to sample $w$ at equally spaced steps in order to exploit
the phase behaviour (\ref{phaseLx}),(\ref{phaseLxtheta}) under scaling.
However,
if we choose to ignore scaling entirely (see point 2 in Section~\ref{sec:motivation}), then
there is no reason to Fourier transform the radial direction, and we introduce an
alternative spectral function,
\bea
	M_\vecx(q,R) &\isdef& \hskip-7mm
	\sum_{a,b\in\{1,\ldots,Z\}\atop a\neq b} \!\!\!\!\!\!\!\!
	e^{iq\qf_{ab}}  \exp\left[-\frac{(R-R_{ab})^2}{2\qs^2}\right]
\label{defM}
	\\
	M_{\vecx\qy}(q,R) 
	\hskip-1mm &\isdef& \hskip-7mm
	\sum_{a,b\in\{1,\ldots,Z\}\atop a\neq b} \!\!\!\!\!\!\!\!
	e^{iq\qf_{ab}}  \exp\left[-\frac{(R-R_{ab})^2}{2\qs^2}\right]
	e^{i(\qy_a-\qy_b)}.
	\;\;\;\;\;\;
\label{defMtheta}
\eea
In the radial direction, the functions $M_\vecx$ and $M_{\vecx\qy}$
consist of a sum of ${Z\choose 2}$ Gaussian peaks centered on the values $R_{ab}$.
The width $\qs>0$ reduces the scheme's sensitivity to small perturbations in the
minutia properties.

Under a rotation ($\vecx_j\mapsto \qF\vecx_j$, $\qy_j\mapsto\qy_j+\qf$) we have
$M_\vecx(q,R)\mapsto e^{iq\qf}M_\vecx(q,R)$ and
$M_{\vecx\qy}(q,R)\mapsto e^{iq\qf}M_{\vecx\qy}(q,R)$.
We want all our spectral functions to be single-valued.\footnote{
Invariant under rotations $\qf$ that are an integer multiple of $2\pi$.
}
Hence $q$ always has to be integer.

\begin{lemma}
\label{lemma:oddvanishes}
For odd $q$ it holds that $L_\vecx(q,w)=0$ for all $w$, and
$M_\vecx(q,R)=0$ for all $R$.
\end{lemma}

\begin{proof}
In (\ref{defL}) every pair of indices $a,b$ gives two terms in the summation.
Using $R_{ba}=R_{ab}$ and $\qf_{ba}\equiv \qf_{ab}+\pi\mod 2\pi$ (see Fig.\,\ref{fig:AngleDistancePair}),
we write $e^{iq\qf_{ab}}e^{iw\ln R_{ab}}$ $+e^{iq\qf_{ba}}e^{iw\ln R_{ba}}$
$=e^{iq\qf_{ab}}e^{iw\ln R_{ab}}[1+e^{iq\pi}]$
$=e^{iq\qf_{ab}}e^{iw\ln R_{ab}}[1+(-1)^q]$. This vanishes when $q$ is odd.
The proof for $M_\vecx$ is analogous.
\end{proof}

\subsection{Choosing the grid points}
\label{sec:grid}

We have to choose a discrete $(q,w)$-grid of points on which to evaluate
$L_\vecx$ and $L_{\vecx\qy}$.
On the one hand, the grid should contain many points so that the 
spectral functions contain sufficient information about the fingerprint.
On the other hand, having too many grid points results in an inefficient scheme.
Lemma~\ref{lemma:oddvanishes} tells us that we do not have to compute $L_\vecx$
for odd~$q$.
Furthermore, we know that, at a given $q$, the spectral functions detect angular 
periodic features of size $\approx 2\pi/q$ radians.
This leads to a natural cutoff at large $q$ where the length scale becomes smaller
than the feature size in a typical fingerprint, and noise starts to dominate.
Similarly, there is a natural maximum for $w$, namely where $2\pi/w$ matches
$\min_{ab: a\neq b}\ln R_{ab}$.
Finally we note that $L_\vecx(-q,-w)=L_\vecx^*(q,w)$
and $L_{\vecx\qy}(-q,-w)=(-1)^q L_{\vecx\qy}^*(q,w)$.
This means that the grid point $(-q,-w)$ contains exactly the same information as $(q,w)$
and hence can be omitted.
The considerations listed above are the only theoretical guidelines for choosing the grid;
the best choice must be found by trial and error.

The considerations for $M_\vecx, M_{\vecx\qy}$ are similar. 
The grid is a $(q,R)$-grid.
The maximum $q$ should be roughly the same as for the $L$-functions.
The natural cutoffs for $R$ are given by $\min_{ab: a\neq b}R_{ab}$ and $\max_{ab} R_{ab}$.
It holds that $M_\vecx(-q,R)=M_\vecx^*(q,R)$ and 
$M_{\vecx\qy}(-q,R)=(-1)^q M_{\vecx\qy}^*(q,R)$.
Hence it suffices to look at positive $q$ only.

\subsection{Introducing weights}
\label{sec:weights}

In the computation of a spectral function at enrollment,
it is possible to introduce a weight factor for each of the $(a,b)$-pairs 
in the summation.
It is advantageous to set a low weight for minutia pairs which are unlikely to be recovered later.
A low recovery likelihood may occur e.g.\;when a minutia has low quality.
Another reason can be a very large value of $R_{ab}$, in which case the recovery is sensitive 
to noise at the edge of the image, or a very small $R_{ab}$ which may cause later minutia misidentification
in case of noise.
In our experiments we have not used weights other than $0$ or~$1$.

\subsection{Choosing the score function}
\label{sec:score}

Let $F$ denote one of the four spectral functions $L_\vecx, L_{\vecx\qy}$, $M_\vecx$, $M_{\vecx\qy}$
obtained at enrollment,
and $F'$ the noisy version of $F$ obtained later, in the verification phase.
We need a metric or `score' function which quantifies how close $F'$ is to~$F$.
As $F$ and $F'$ are complex-valued, there are quite some options.
We have experimented with correlation functions for the radial and phase part of the complex numbers,
as well as the real and imaginary part. Furthermore we have tried distance in the complex plane,
with and without normalisation of the function $F$ as a whole.
In our experiments
it turns out that a complex correlation-like quantity
is best able to discriminate between genuine fingerprint
matches and impostors. We define our score $S$ as
\be
	S(F,F')=|\qr(F,F')|
\label{score}
\ee
where $\qr$ stands for the correlation as defined in Section~\ref{sec:notation},
and the matrices $F,F'$ are treated as vectors.

\subsection{Fusion of scores}
\label{sec:fusion}

The spectral functions $L_\vecx$ and $L_{\vecx\qy}$ together contain more information
about the fingerprint than each one separately.
The information is partially overlapping.
We construct a `fused' score by adding the two scores (\ref{score}) in the same way as
\cite{XVBKAG2009}:
$S(L_\vecx,L_\vecx')+S(L_{\vecx\qy},L_{\vecx\qy}')$. 
Analogously, for the $M$-functions we work with the fused score
$S(M_\vecx,M_\vecx')+S(M_{\vecx\qy},M_{\vecx\qy}')$. 

\section{Experimental results}
\label{sec:results}

We have applied our minutia-pair approach to 
the Verifinger database 
and the MCYT database \cite{MCYT}.
The Verifinger database contains fingerprints from 
six individual persons, ten fingers per individual, eight images per finger.
The size of each image is $326\times 357$ pixels.
The MCYT database
contains fingerprints from 100 individuals, 10 fingers per individual, 12 images per finger ($256\times 400$ pixels).
The fingerprints are generally of higher quality than in the Verifinger database.

We extracted 
minutia coordinates and orientations from the images by using the VeriFinger software \cite{VeriFinger}.
\subsection{Optimal parameter choices}
\label{sec:optparam}

Good results were obtained with the following parameter settings.
For the $L$-functions, $|q|\in\{1,\ldots,24\}$ and
$w\in[0.2,37.7]$ with 32 equally spaced values.
For the $M$-functions, $q\in\{1,\ldots,16\}$;
$R\in[16,130]$ with 20 equally spaced points (MCYT database);  
$R\in[16,160]$ with 25 equally spaced points (Verifinger database).
For the $L_\vecx$ and $M_\vecx$ functions we take only even $q$, as explained in Lemma~\ref{lemma:oddvanishes}.
We set $\qs=2.3$ pixels.

A minutia extracted by VeriFinger is labeled with a quality $Q \in [0, 100]$. 
We took only minutiae with $Q\geq 45$.
Furthermore we used an additional selection rule that turns out to improve overall results a bit:
a minutia pair is discarded from the $\sum_{ab}$ summation in (\ref{defL},\ref{defLtheta},\ref{defM},\ref{defMtheta}) if $2R_{ab}$ exceeds
the horizontal size of the image.

In Fig.\,\ref{fig:visualM} we show an example of the $M_\vecx$ and $M_{\vecx\qy}$ spectral function.
Entirely different fingers obviously produce very different results.
The two leftmost columns correspond to the same finger.
Noisy images of the same finger do not produce results that, 
to the human eye,
are clearly correlated.
However, it turns out (Section~\ref{sec:ROC}) that the similarities are enough to
distinguish between the enrolled user from an impostor. 

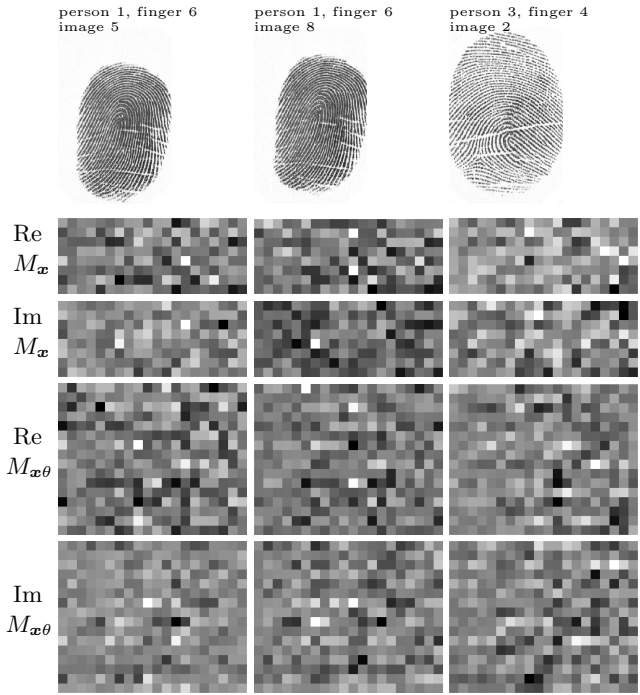
\begin{figure}[h]
\begin{center}
\setlength{\unitlength}{1mm}
\begin{picture}(80,90)(0,0)

\put(6,65){\includegraphics[width=15mm]{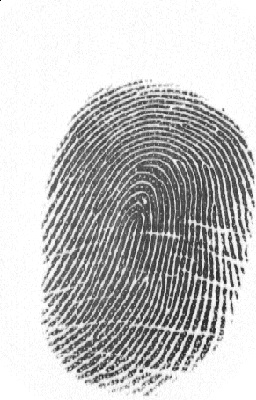}}
\put(32,65){\includegraphics[width=15mm]{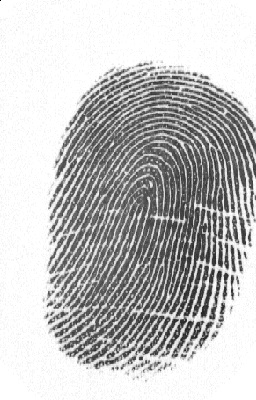}}
\put(58,65){\includegraphics[width=15mm]{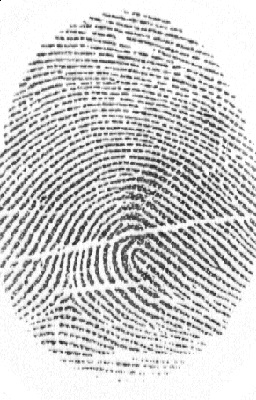}}
\put(6,90){\tiny person 1, finger 6}
\put(6,88){\tiny image 5}
\put(32,90){\tiny person 1, finger 6}
\put(32,88){\tiny image 8}
\put(58,90){\tiny person 3, finger 4}
\put(58,88){\tiny image 2}

\put(0,60){\small $\Re$}
\put(0,56){\small $M_\vecx$}
\put(6,53){\includegraphics[width=25mm]{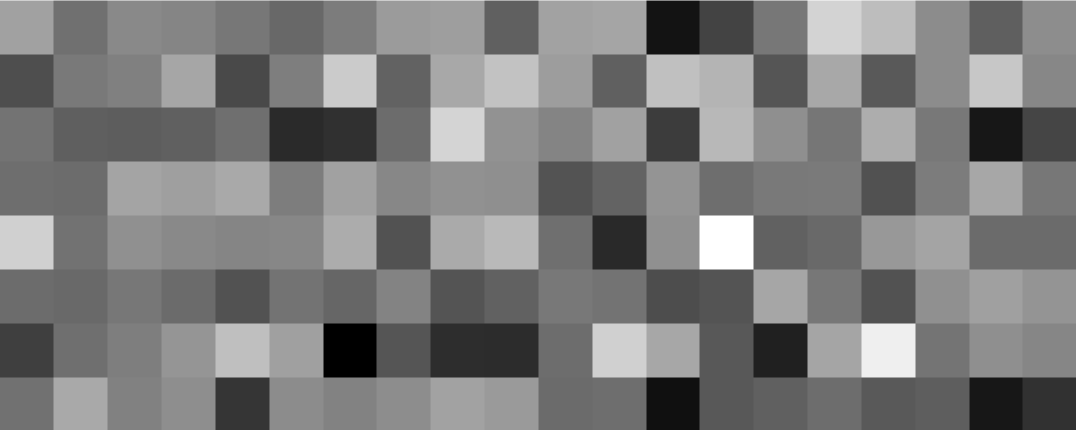}}
\put(32,53){\includegraphics[width=25mm]{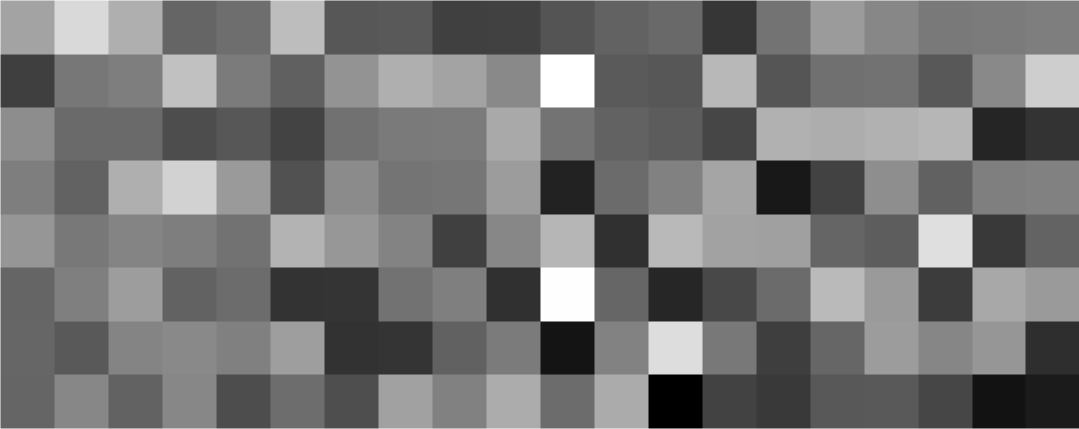}}
\put(58,53){\includegraphics[width=25mm]{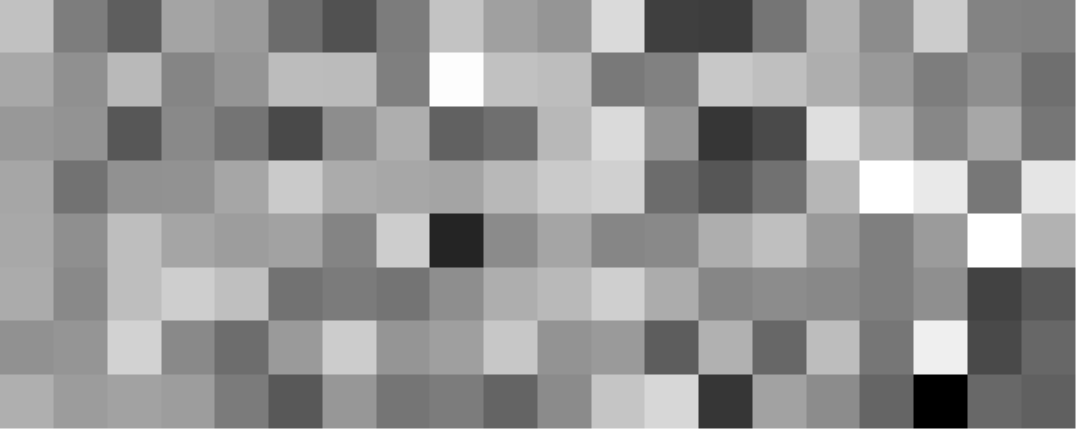}}

\put(0,49){\small $\Im$}
\put(0,45){\small $M_\vecx$}
\put(6,42){\includegraphics[width=25mm]{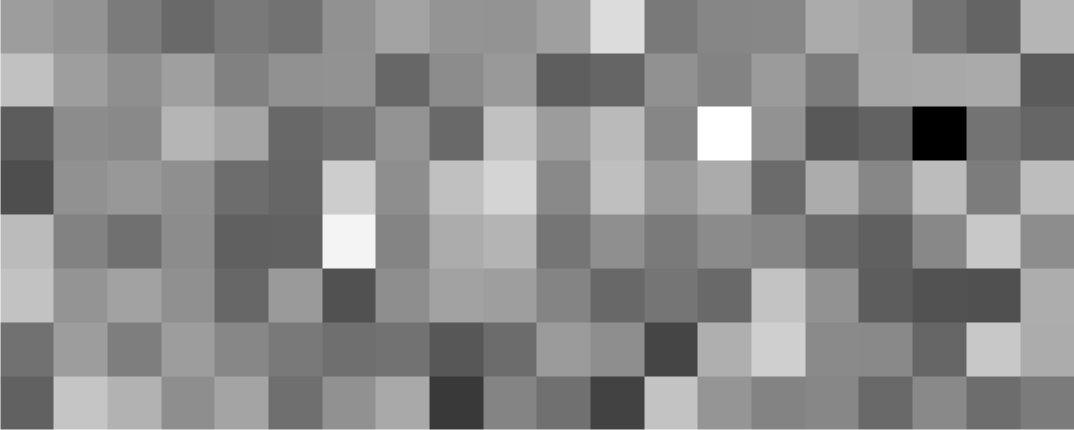}}
\put(32,42){\includegraphics[width=25mm]{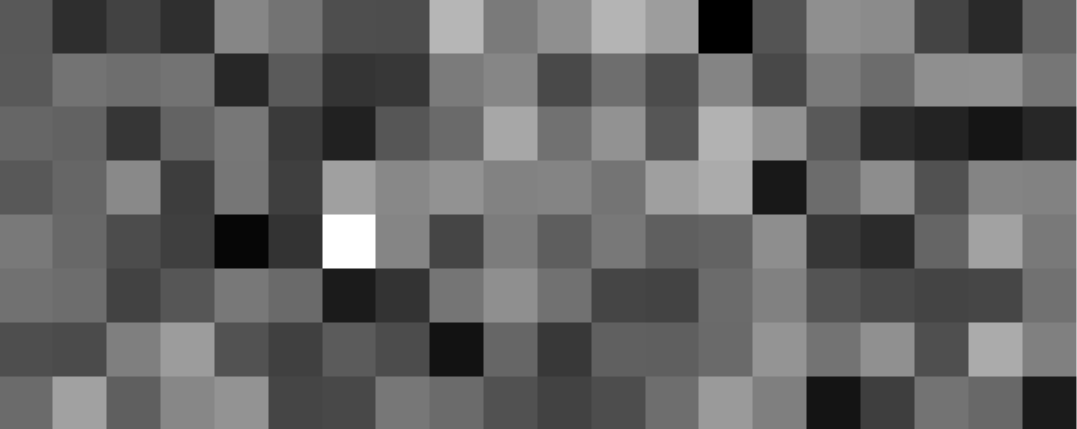}}
\put(58,42){\includegraphics[width=25mm]{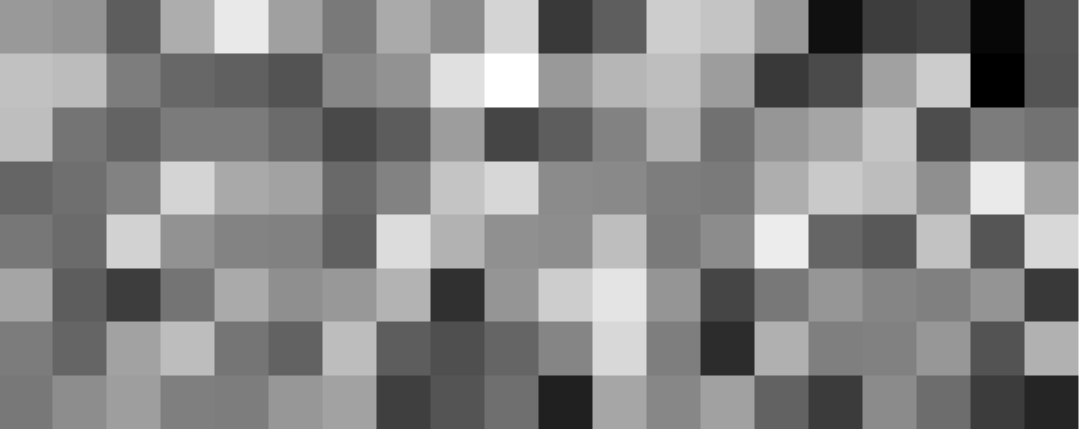}}

\put(0,33){\small $\Re$}
\put(-1,29){\small $M_{\vecx\qy}$}
\put(6,21){\includegraphics[width=25mm]{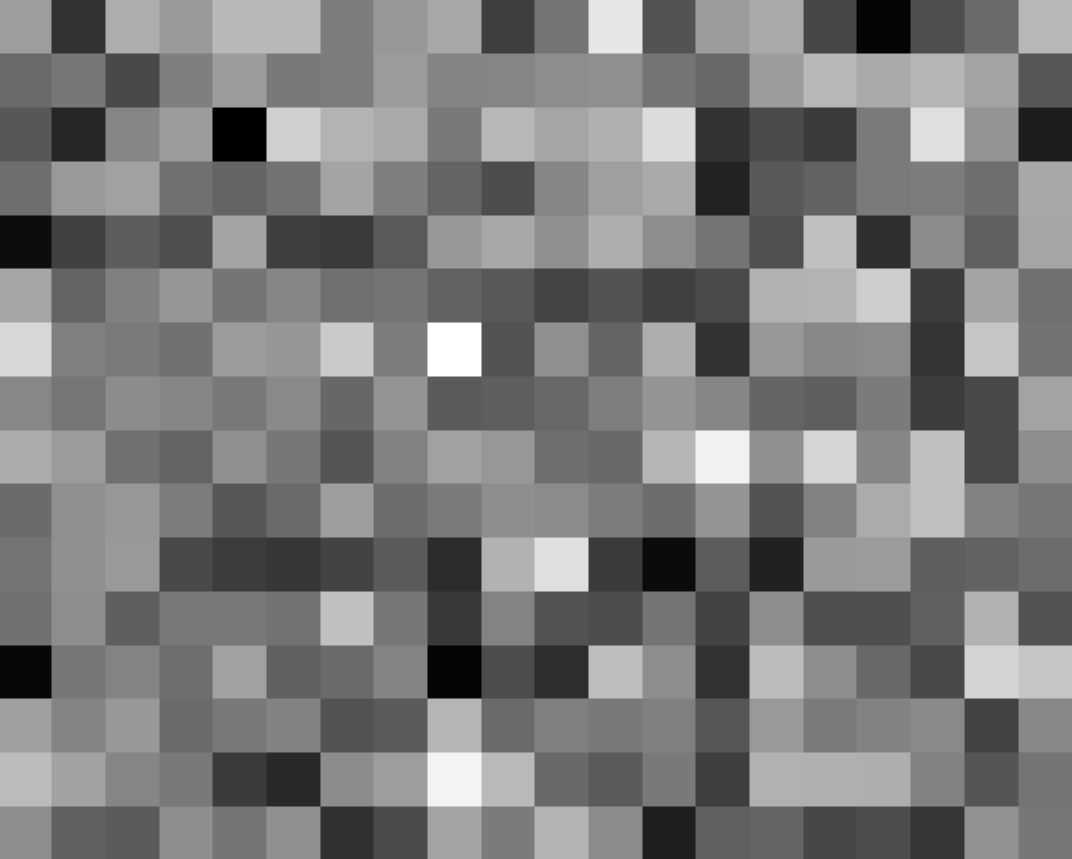}}
\put(32,21){\includegraphics[width=25mm]{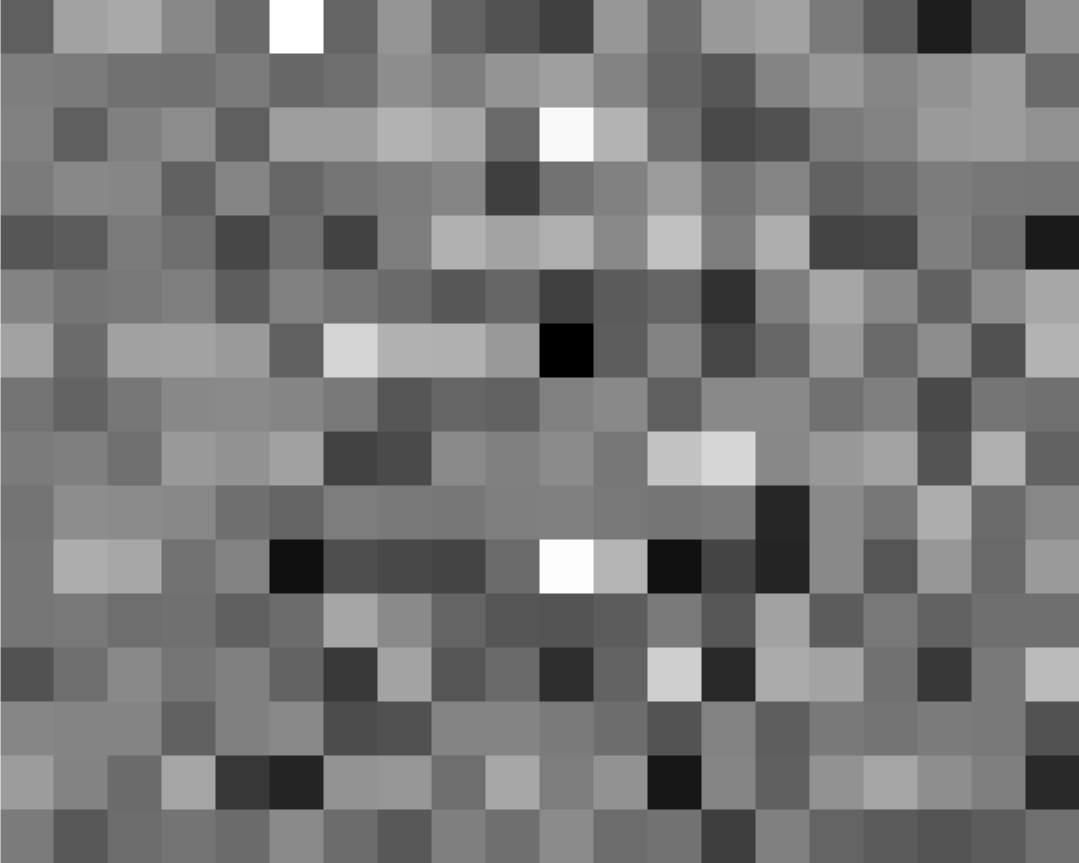}}
\put(58,21){\includegraphics[width=25mm]{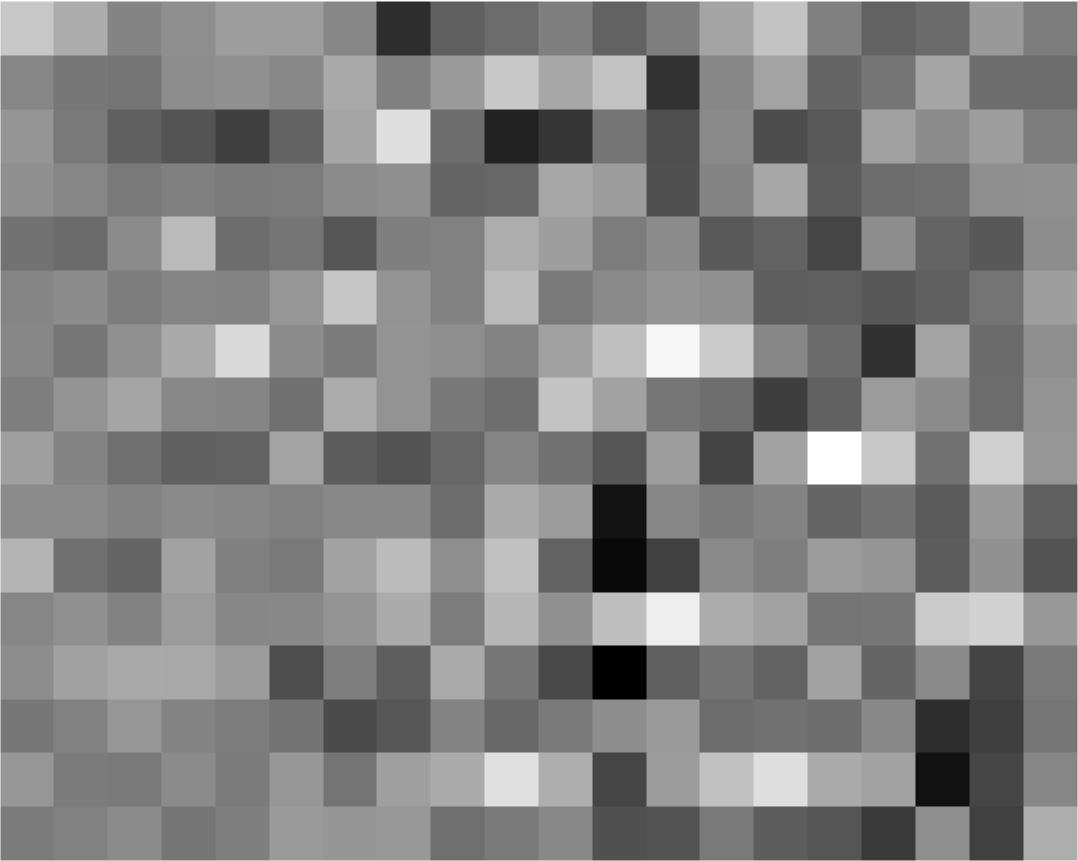}}

\put(0,12){\small $\Im$}
\put(-1,8){\small $M_{\vecx\qy}$}
\put(6,0){\includegraphics[width=25mm]{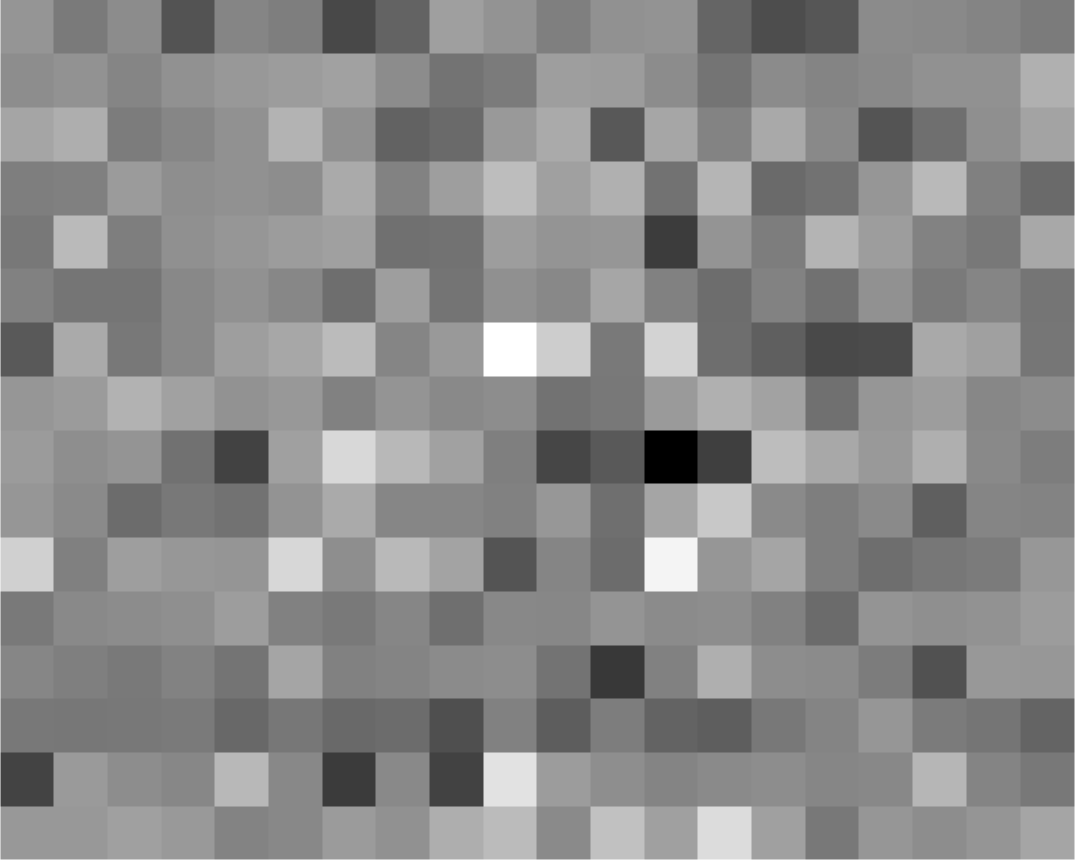}}
\put(32,0){\includegraphics[width=25mm]{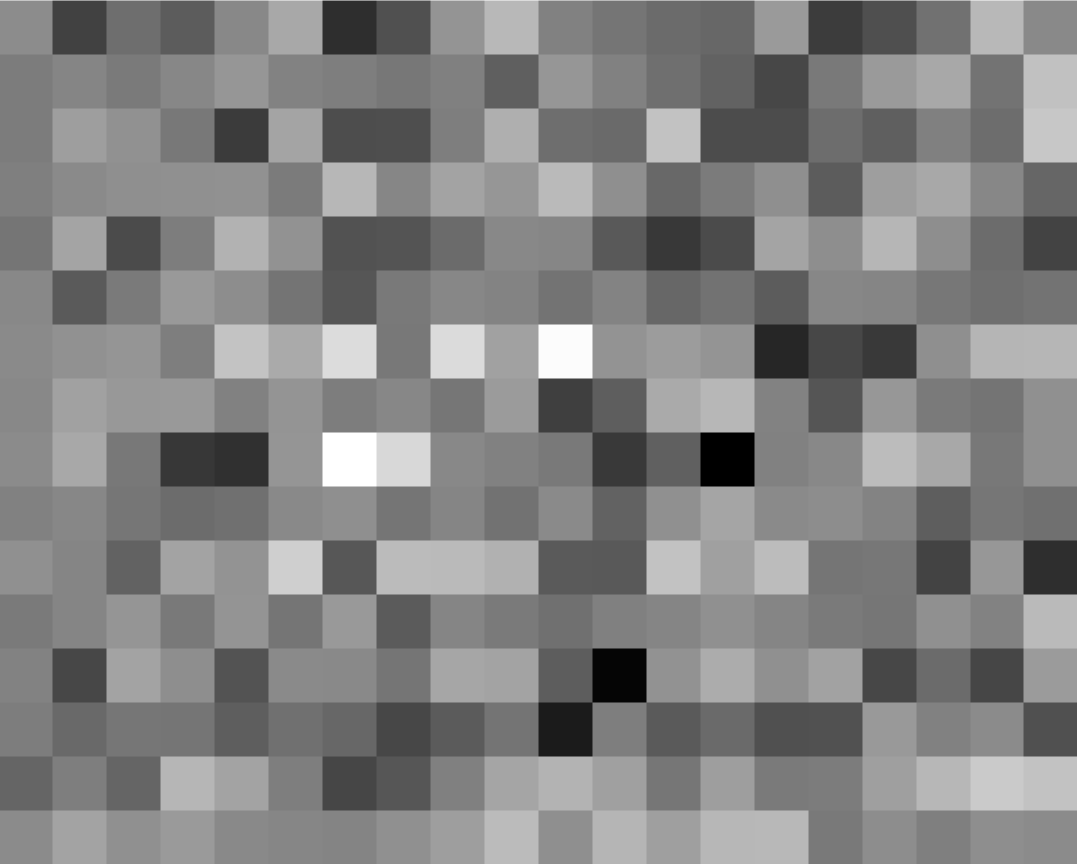}}
\put(58,0){\includegraphics[width=25mm]{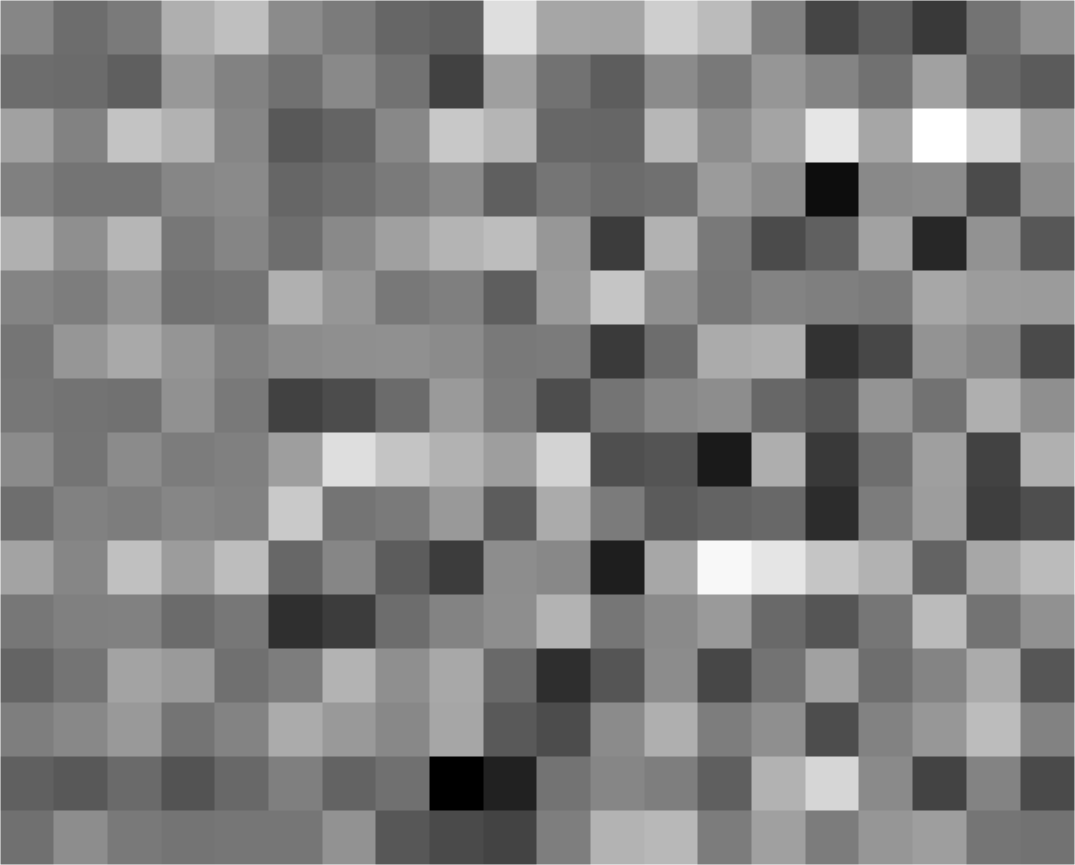}}

\end{picture}
\caption{\it Example of the spectral functions $M_\vecx$ and $M_{\vecx\qy}$.
MCYT database. 
The vertical axis is the $q$-axis, with $q$ increasing upward. 
In each image, black represents the most negative value on the grid,
and white the most positive.}
\label{fig:visualM}
\end{center}
\end{figure}

\subsection{ROC curves and Equal Error Rates}
\label{sec:ROC}

We work in a {\em verification} setting, i.e. a stated identity has to be verified.
We determine the False Rejection Rate (FRR) by comparing, for each finger in the database,
all the pairs of images.
We determine the False Acceptance Rate (FAR) by looking at each pair of different fingers,
where one image is drawn at random for each finger (independently per pair).\footnote{
This includes pairs of unlike fingers, e.g.\;thumb vs index finger.
The statistics do not change much when only pairs of like fingers are compared.
}
We draw Receiver Operating Characteristic (ROC) curves as FAR plotted against FRR. 
Each point in the ROC curve corresponds to one threshold setting.
The Equal Error Rate (EER) is the error rate in the point where FRR equals FAR. 

Table~\ref{t:EER} lists the EER values that we obtained.
The ROC curves are shown in Fig.\,\ref{fig:ROC}.
We see that the $M$-functions consistently outperform the $L$-functions,
and that the $L_{\vecx\qy}, M_{\vecx\qy}$ spectral functions outperform the
location based functions.
Furthermore we see that fusion of $M_\vecx$ and $M_{\vecx\qy}$
yields only a modest improvement over~$M_{\vecx\qy}$.
We conclude that, in our pair-based approach, the best option is to work either with
$M_{\vecx\qy}$ or the fusion of $M_\vecx$ and $M_{\vecx\qy}$.

We benchmark our system against results reported by Xu et al. \cite{XVBKAG2009},
which are based on ten individuals in the MCYT database who have high-quality fingerprint images.
The ROC curves are shown in Fig.\;\ref{fig:ROC10}, and
Table~\ref{t:benchmark} contains the EER comparison.\footnote{
Unfortunately, \cite{XVBKAG2009} does not mention which ten individuals were selected.
}
We conclude that our pair-based spectral function $M_{\vecx\qy}$ has a discrimination performance comparable to
Xu et al.'s spectral function.

\begin{table}[h]
\begin{center}
\caption{\it Equal Error Rates obtained with the parameter settings given in Section~\ref{sec:optparam}.
The notation `$F$' stands for either $L$ or $M$.
No rotation of the verification image.}
\label{t:EER}
\begin{tabular}{|c|c|c|c|c|}
\hline
Database & Function $F$ & $F_\vecx$ & $F_{\vecx\qy}$ & Fusion
\\ \hline\hline
MCYT & $L$ & 5.3\% & 3.5\% & 3.0\%
\\ \hline
& $M$ &4.0\% & 2.5\% & 2.2\%
\\ \hline\hline
Verifinger & $L$ & 11\% & 4.9\% & 5.7\%
\\ \hline
& $M$ & 8.0\% & 3.3\% & 3.2\%
\\ \hline
\end{tabular}
\end{center}
\end{table}

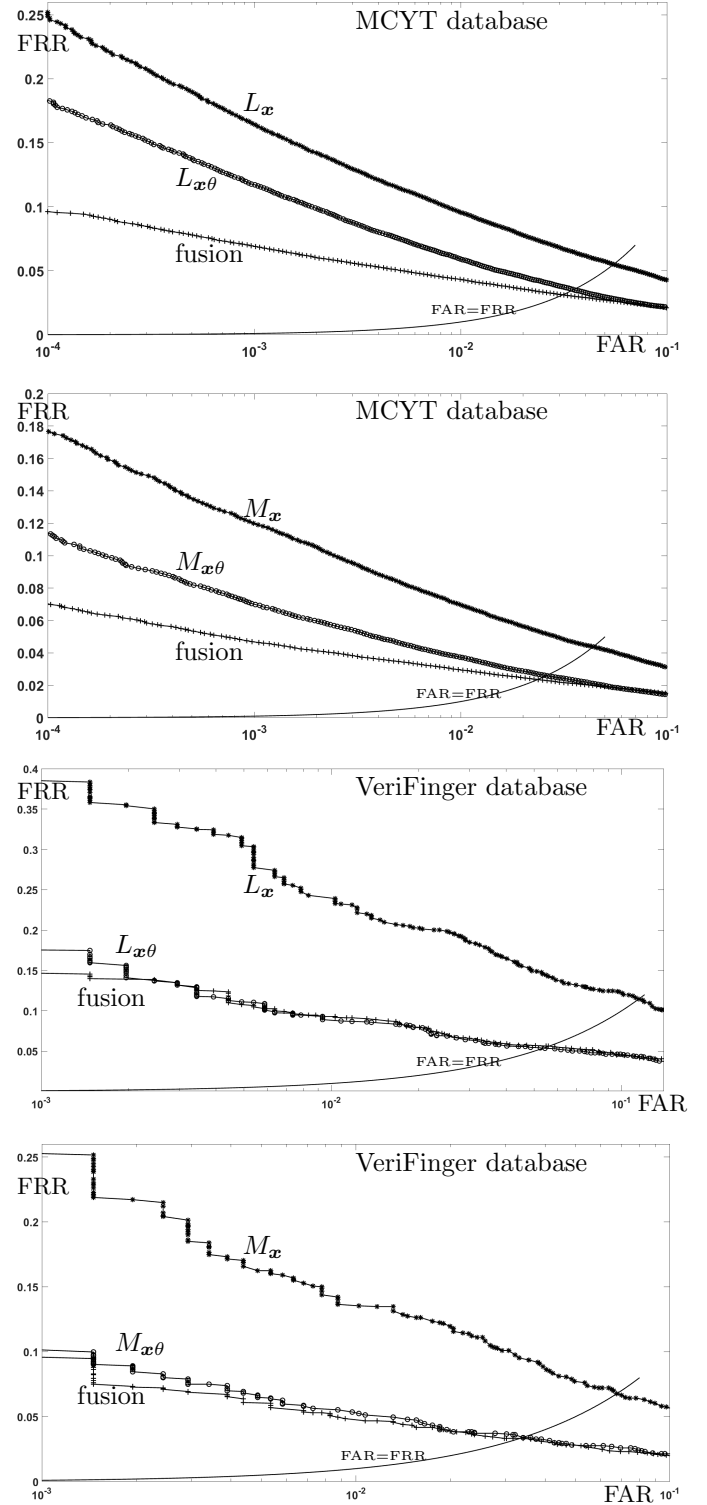
\begin{figure}[t]
\begin{center}
\setlength{\unitlength}{1mm}
\begin{picture}(90,200)(0,0)

\put(0,152){\includegraphics[width=88mm]{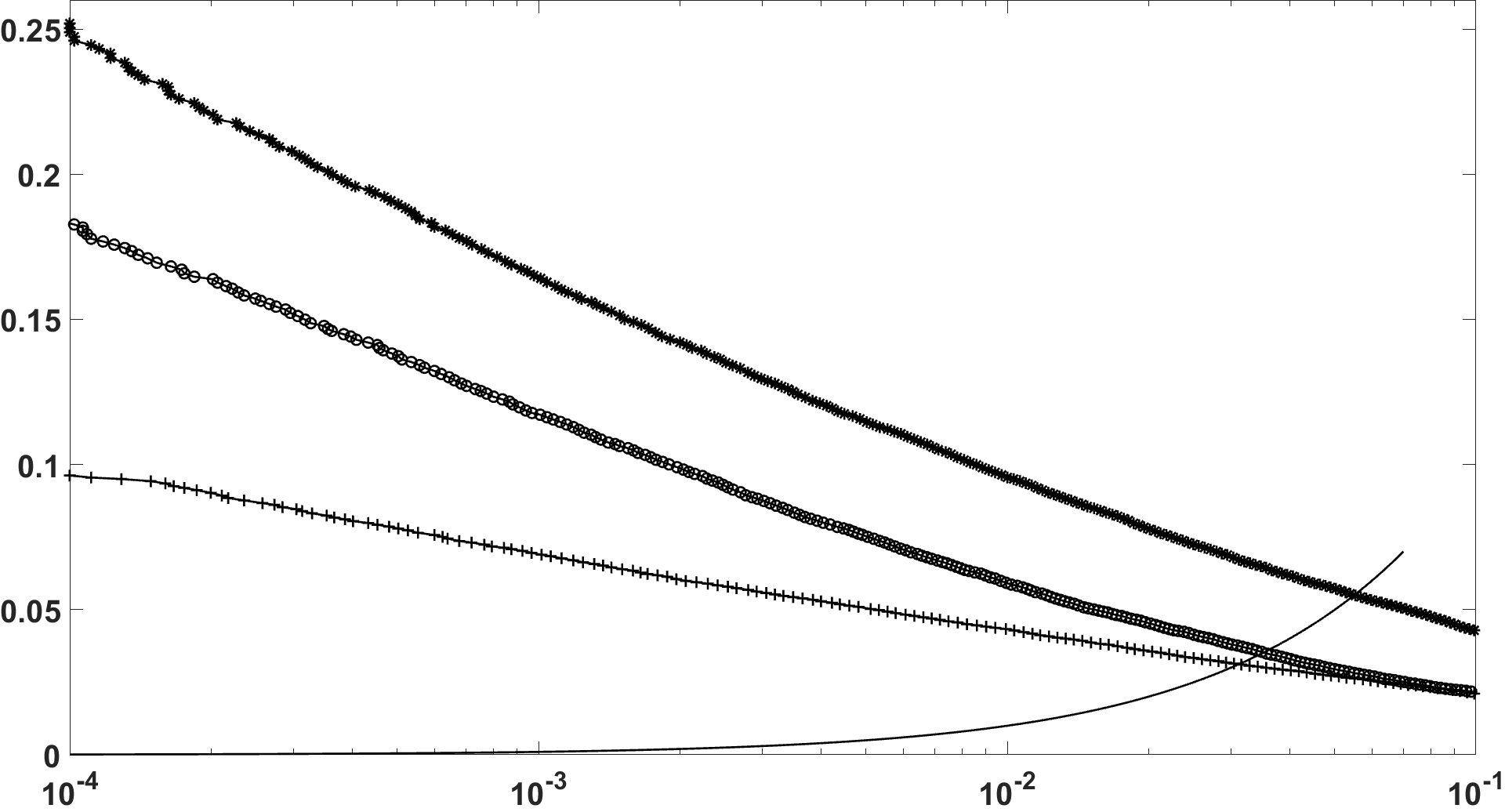}}
\put(45,196){MCYT database}
\put(30,185){$L_\vecx$}
\put(21,175){$L_{\vecx\qy}$}
\put(21,165){fusion}
\put(55,158){\tiny FAR=FRR}
\put(77,153){\small FAR}
\put(0,193){\small FRR}

\put(0,101){\includegraphics[width=88mm]{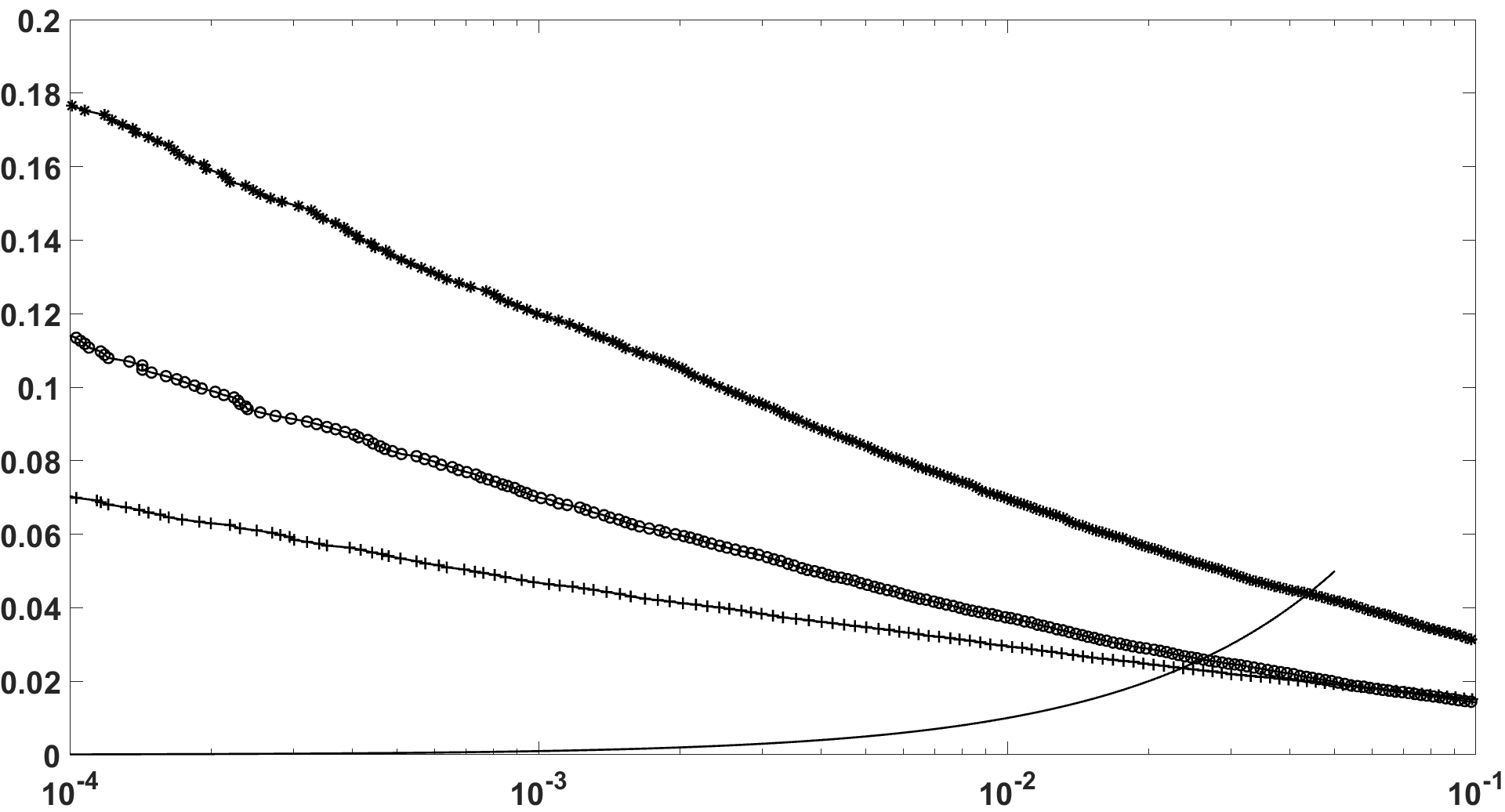}}
\put(45,144){MCYT database}
\put(30,131){$M_\vecx$}
\put(21,124){$M_{\vecx\qy}$}
\put(21,112){fusion}
\put(53,107){\tiny FAR=FRR}
\put(77,102){\small FAR}
\put(0,144){\small FRR}

\put(0,52){\includegraphics[width=86mm]{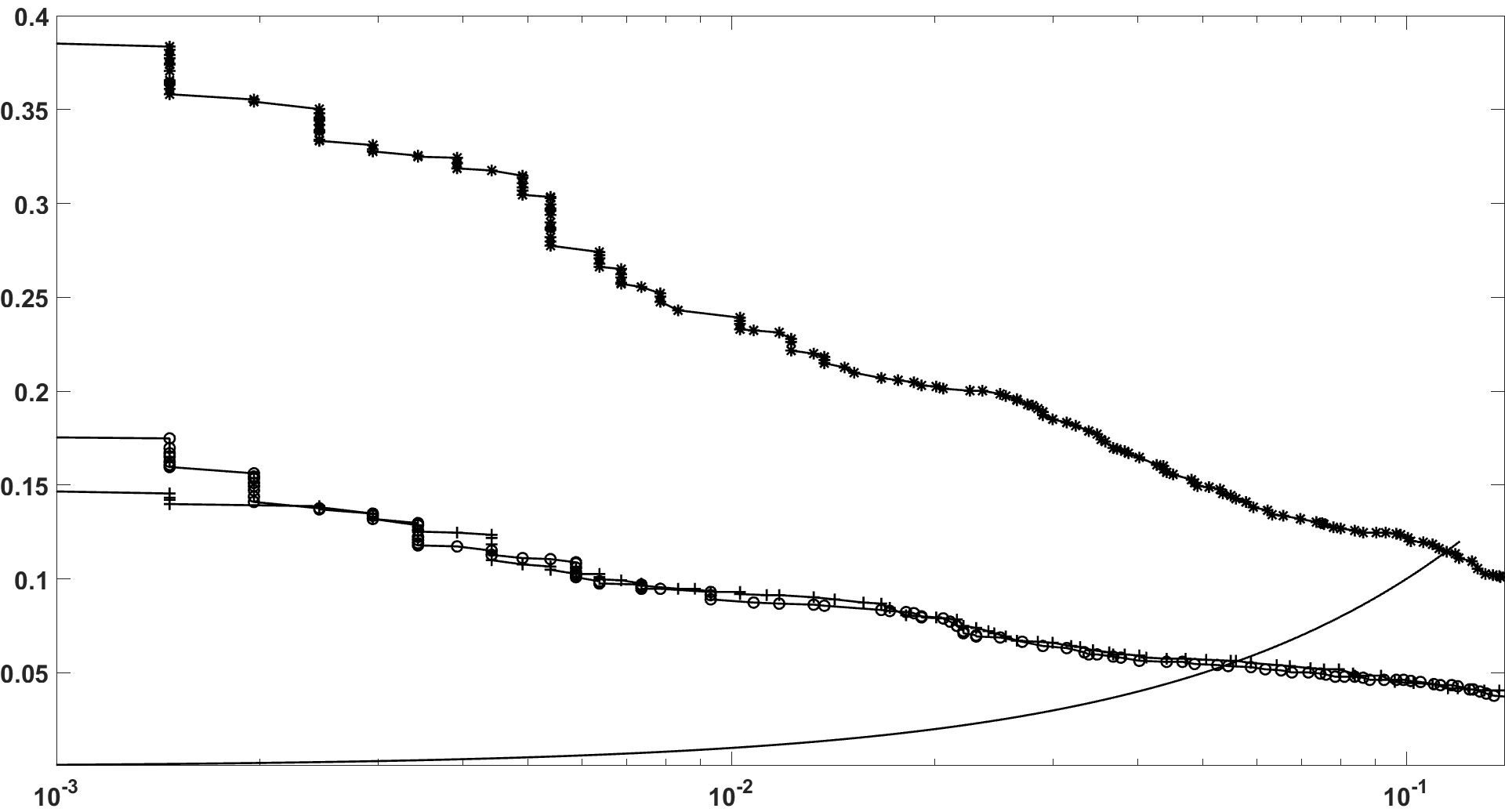}}
\put(45,94){VeriFinger database}
\put(30,81){$L_\vecx$}
\put(13,73){$L_{\vecx\qy}$}
\put(8,66){fusion}
\put(53,58){\tiny FAR=FRR}
\put(82.5,52){\small FAR}
\put(0,93.5){\small FRR}

\put(0,0){\includegraphics[width=88mm]{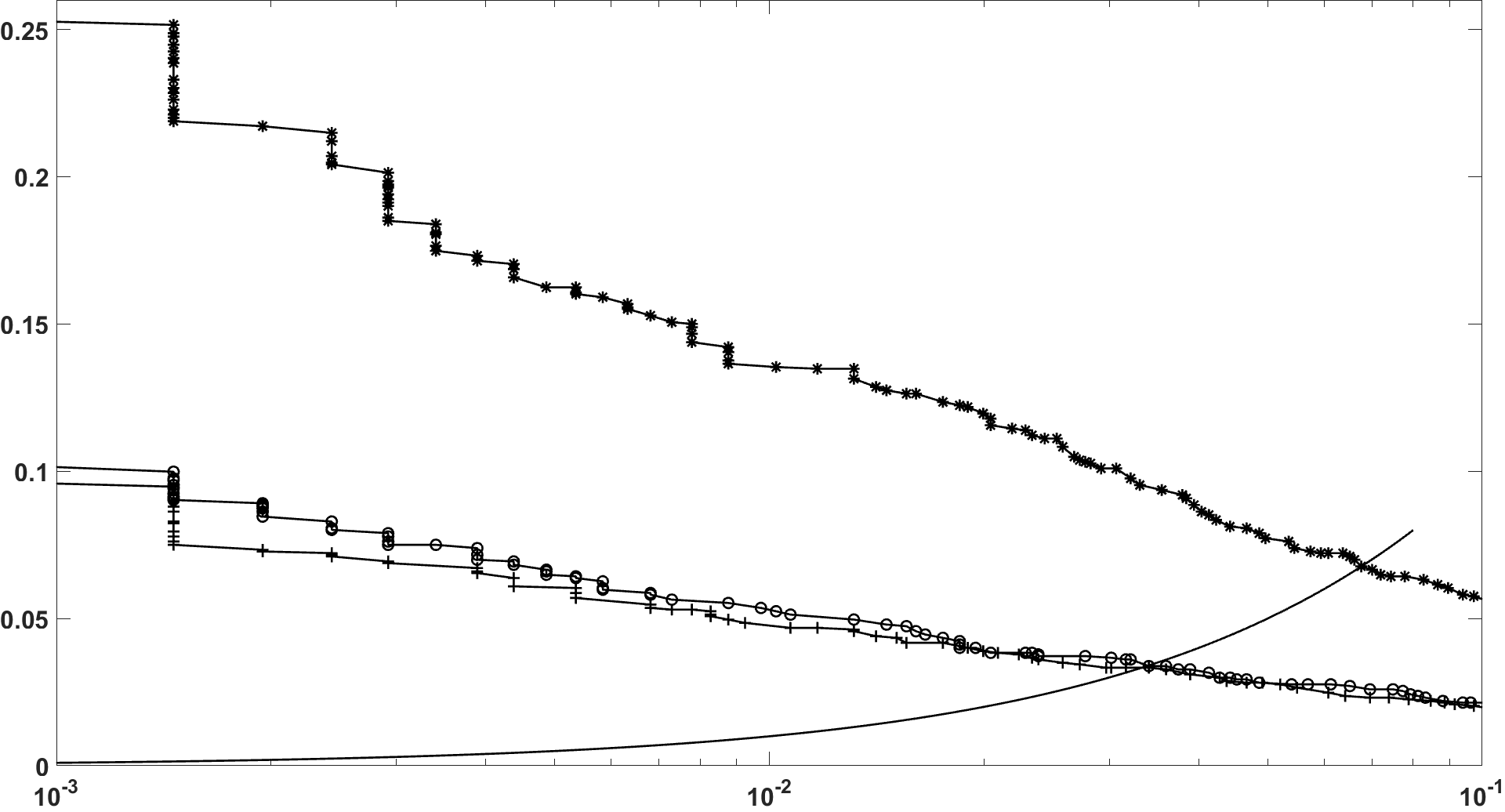}}
\put(45,44){VeriFinger database}
\put(30,33){$M_\vecx$}
\put(13,20){$M_{\vecx\qy}$}
\put(8,13){fusion}
\put(43,5.5){\tiny FAR=FRR}
\put(78,0){\small FAR}
\put(0,41){\small FRR}

\end{picture}
\caption{\it ROC curves for our pair-based spectral functions applied to two databases.
No rotation of the verification image.}
\label{fig:ROC}
\end{center}
\end{figure}


\begin{table}[h]
\caption{\it Equal Error Rates for a subset of ten individuals in the MCYT database
who have high-quality fingerprints. 
No rotation of the verification image.
The last row is from Table~VI in \cite{XVBKAG2009}.
The $L$ and $M$ function were computed for individuals 16,24,26,32,34,35,46,53,80,94.
}
\label{t:benchmark}
\begin{center}
\begin{tabular}{|l|c|c|c|}
\hline
Function $F$ & $F_\vecx$ & $F_{\vecx\qy}$ & Fusion
\\ \hline\hline
$L$ & 1.1\% & 0.73\% & 0.31\%
\\ \hline
$M$ & 0.65\% & 0.35\% & 0.15\%
\\ \hline
Xu et al & 0.47\%  & 0.42\%  & 0.22\% 
\\ \hline
\end{tabular}
\end{center}
\end{table}

\begin{figure}
\begin{center}
\setlength{\unitlength}{1mm}
\begin{picture}(90,100)(0,0)

\put(0,52){\includegraphics[width=88mm]{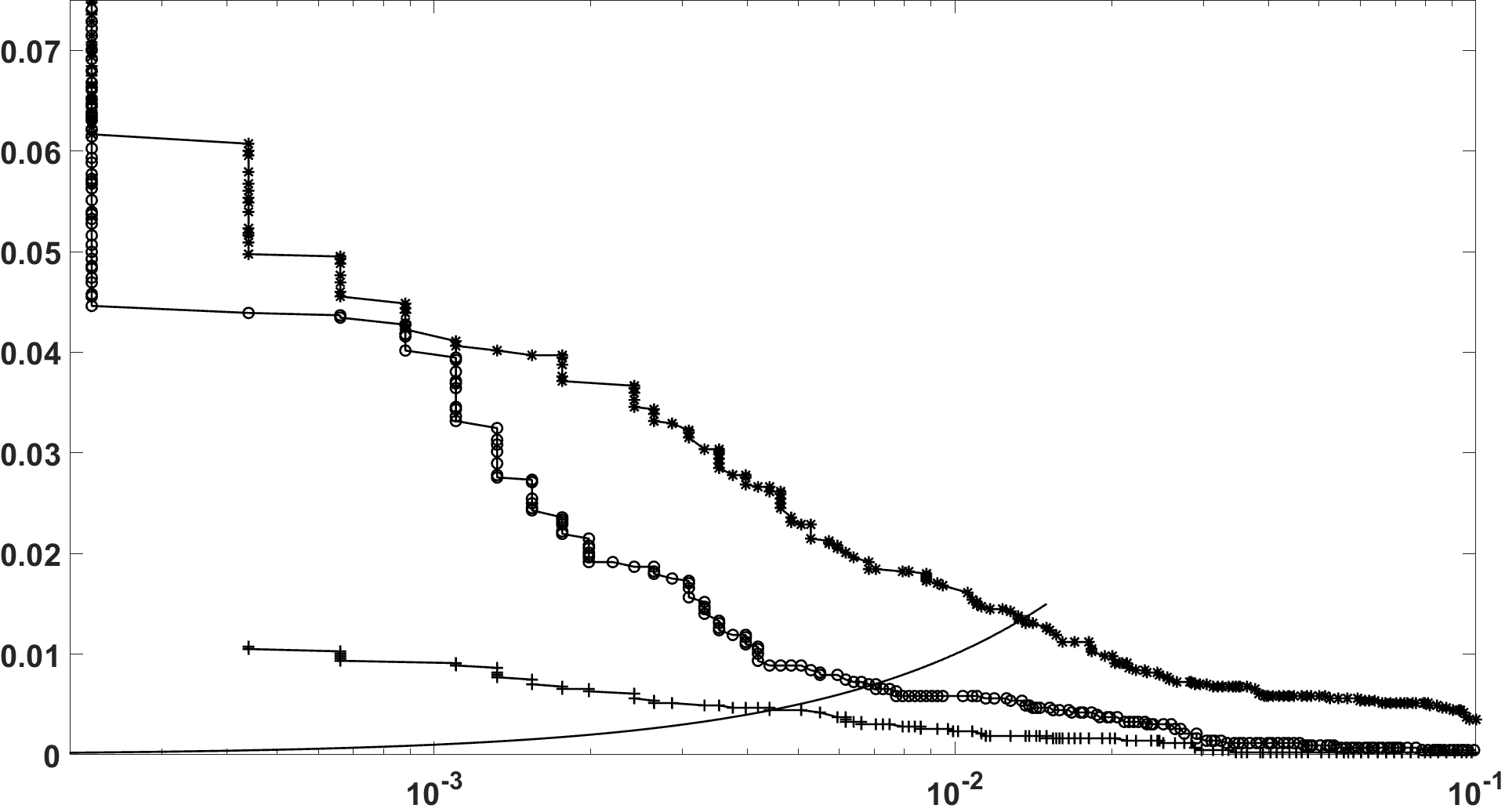}}
\put(35,96){MCYT database, 10 persons}
\put(32,80){$L_\vecx$}
\put(10,78){$L_{\vecx\qy}$}
\put(12,63){fusion}
\put(62,65){\tiny FAR=FRR}
\put(77,52){\small FAR}
\put(6,96){\small FRR}

\put(0,01){\includegraphics[width=88mm]{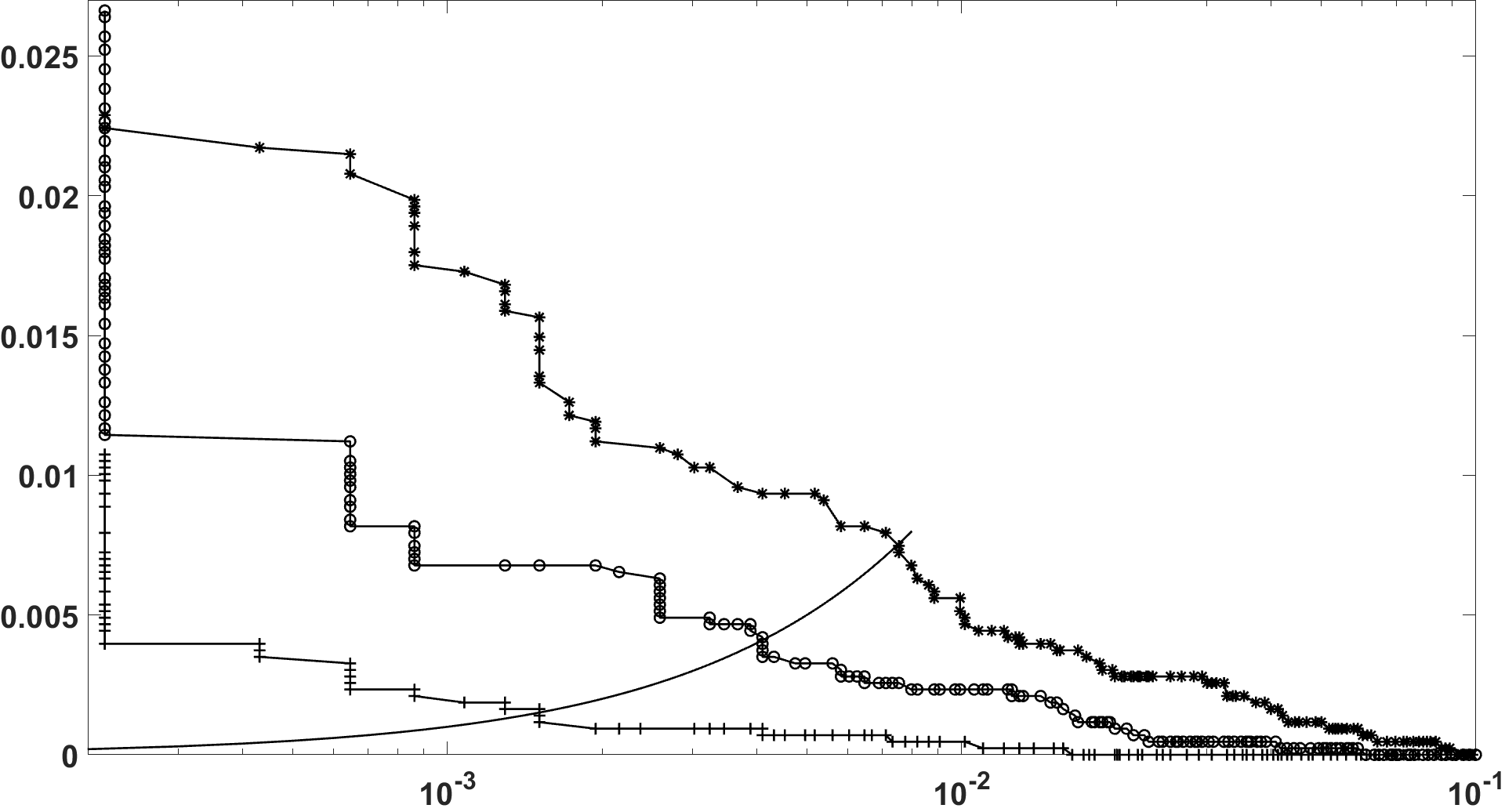}}
\put(35,45){MCYT database, 10 persons}
\put(30,31){$M_\vecx$}
\put(12,24){$M_{\vecx\qy}$}
\put(10,12){fusion}
\put(54,18){\tiny FAR=FRR}
\put(77,1){\small FAR}
\put(7,45){\small FRR}

\end{picture}
\caption{\it ROC curves for the ten-person subset of the MCYT database.
No rotation of the verification image.}
\label{fig:ROC10}
\end{center}
\end{figure}

\subsection{Rotation of the verification image}
\label{sec:imagerotation}

The results of Section~\ref{sec:ROC} were obtained without 
Xu et al.'s procedure of trying out several image rotations 
so as to optimise the matching score.
Now we discuss what happens when we do try a number of different rotation angles~$\qf$.

First we checked for the MCYT and the VeriFinger database how a rotation 
$\qf\in(-10^\circ,+10^\circ)$
affects the $M_\vecx$ and $M_{\vecx\qy}$-based score in case of a genuine image pair.
At some optimal angle $\qf_0$ the score is maximal.
For all genuine pairs we determined $\qf_0$, for $M_\vecx$ and $M_{\vecx\qy}$.
The histograms of $\qf_0$ are shown in Fig.\;\ref{fig:rothist}.
We see that typically $|\qf_0|< 6^\circ$.

In Fig.\,\ref{fig:rotROC} we present ROC curves that
show the impact of trying multiple rotation angles $\qf$ in a limited range;
we set the range based on Fig.\,\ref{fig:rothist}.
In the case of the MCYT database we see a consistent though small improvement.
For the VeriFinger database the change is not always favourable;
the ROC curves intersect. 
For both databases, the effect on the EER is minimal.

Increasing the range of $\qf$ does not improve the matching of genuine pairs;
it does however increase the FAR. Hence the ROC curves become worse when we increase 
the range of~$\qf$.

These results allow for a very interesting trade-off:
instead of opting for a minimal improvement of matching accuracy,
we can skip the $\qf$-search and thus significantly reduce the computation time.
Note that Xu et al.'s method has a $\qf$-search with 11 different values of~$\qf$.

\begin{figure}[t]
\begin{center}
\setlength{\unitlength}{1mm}
\begin{picture}(90,50)(0,0)

\put(0,25){\includegraphics[width=44mm]{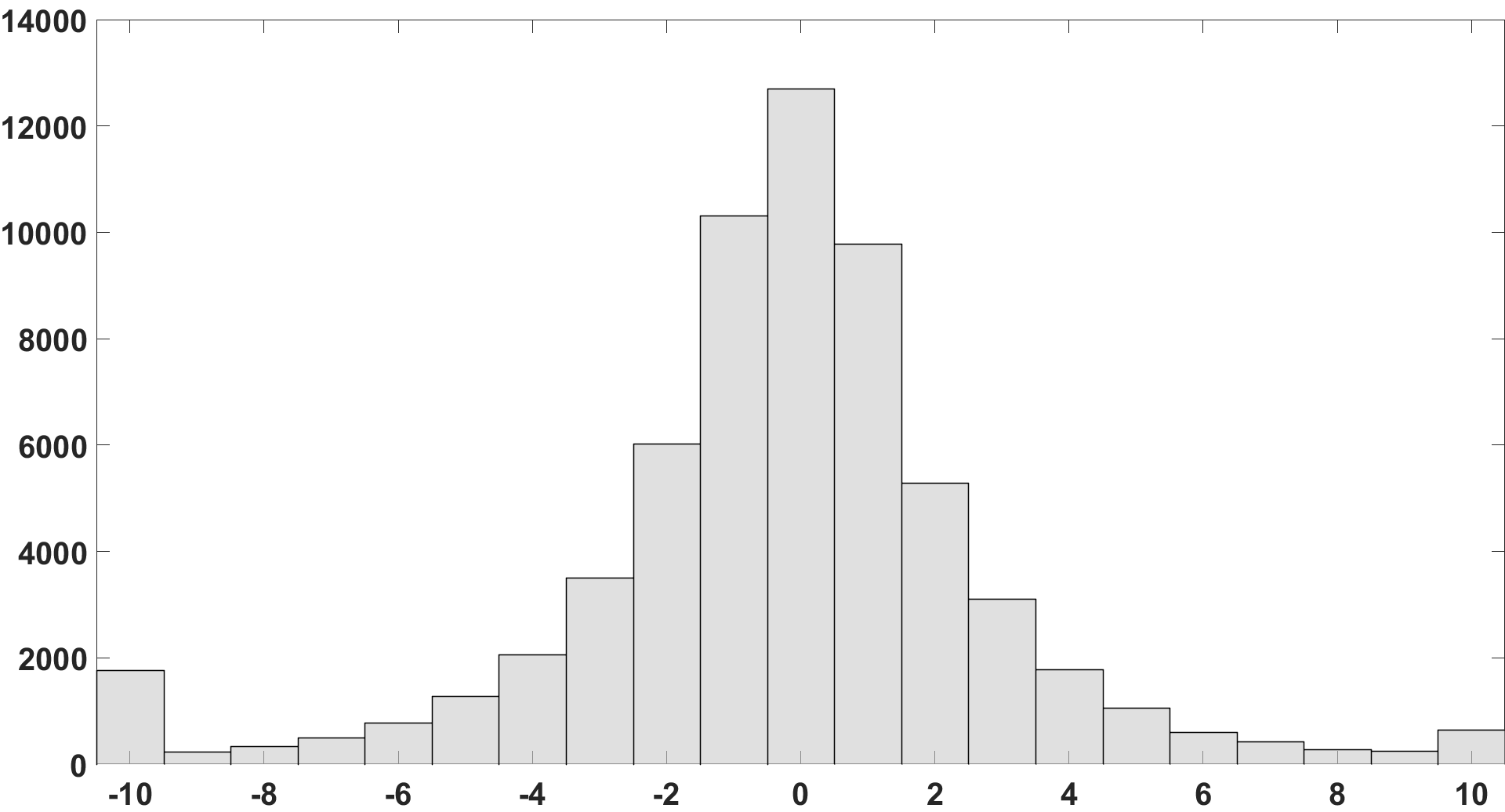}}
\put(44,25){\includegraphics[width=45mm]{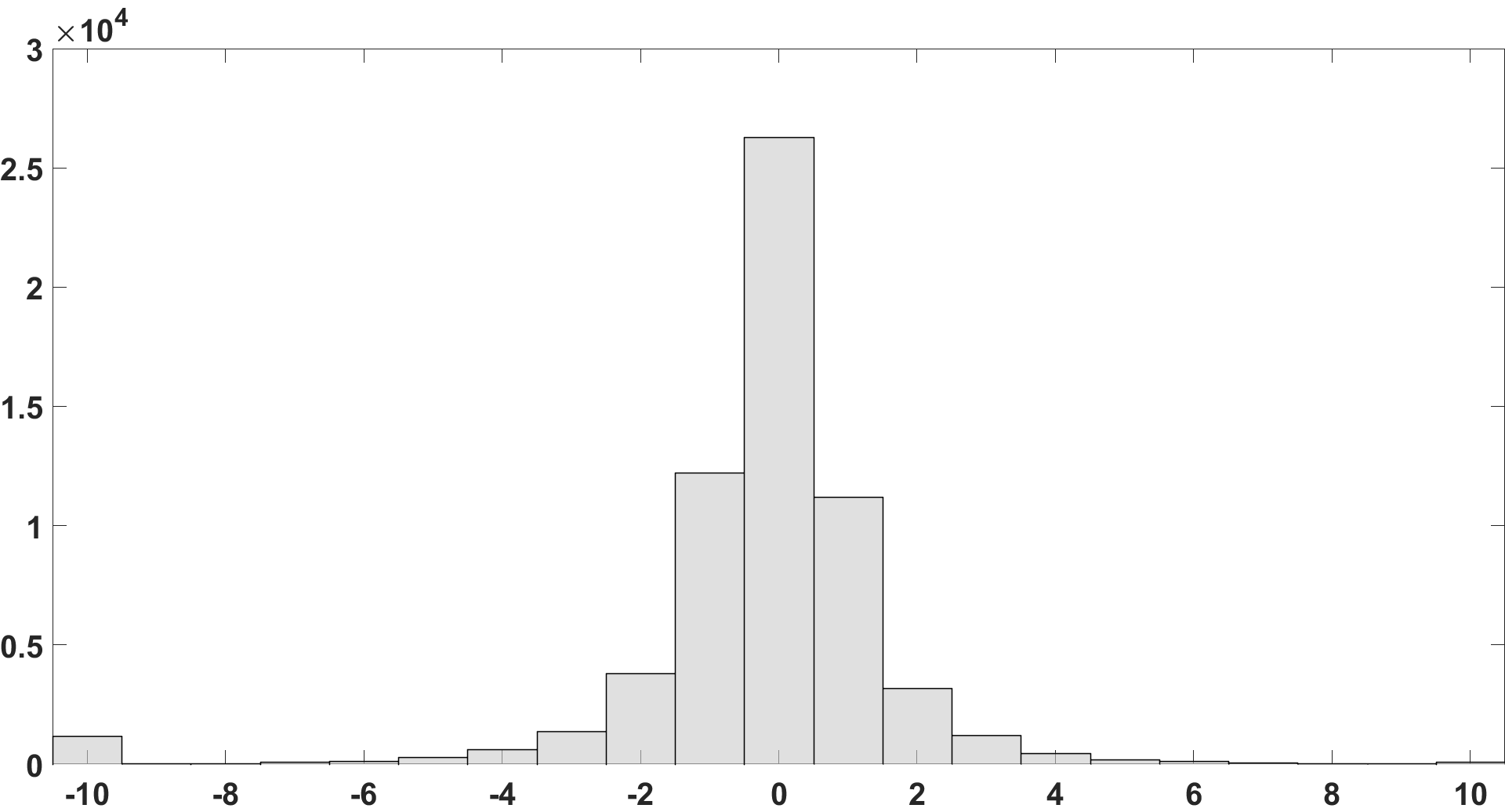}}
\put(0.5,0){\includegraphics[width=43mm]{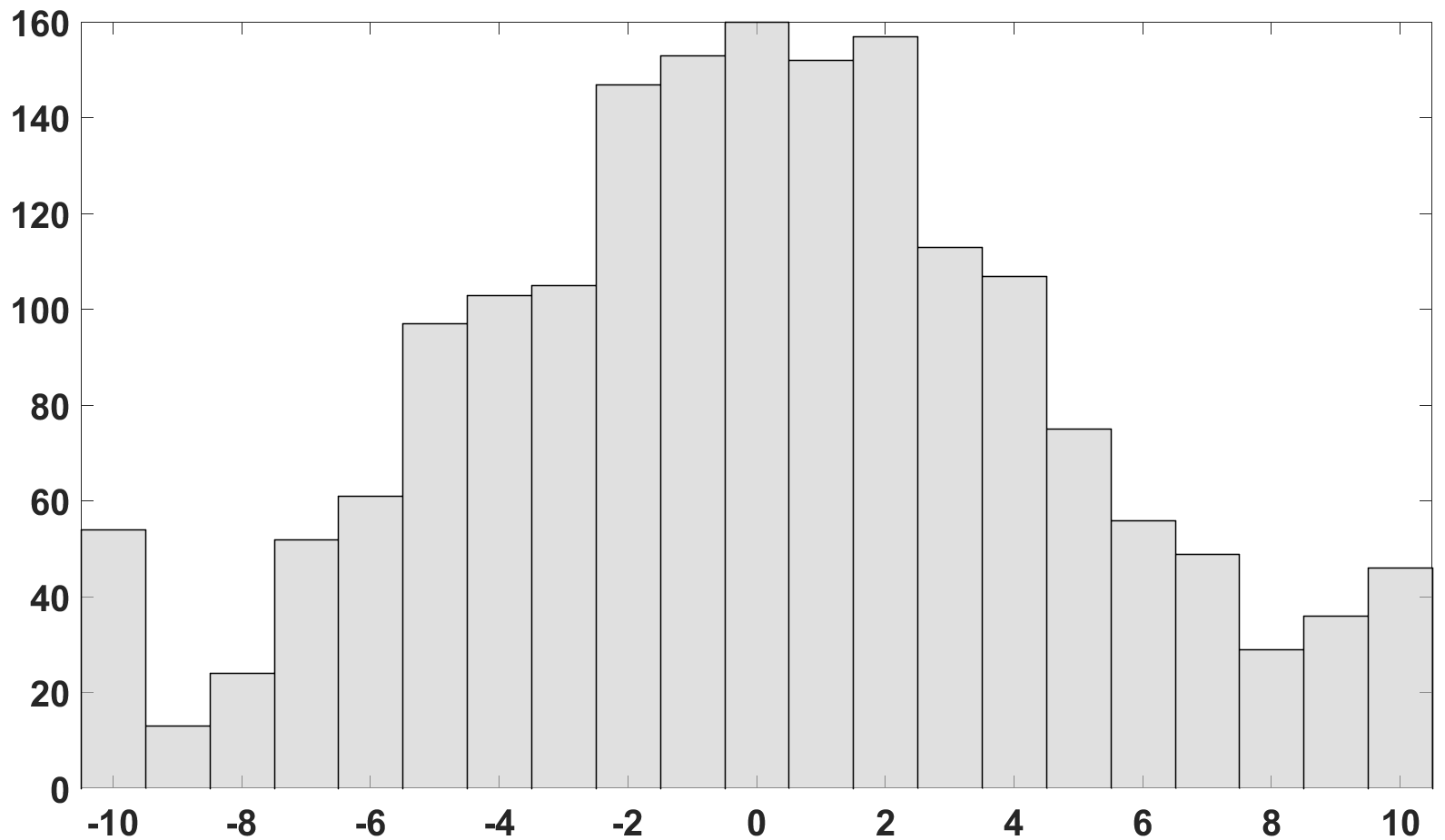}}
\put(44,0){\includegraphics[width=45mm]{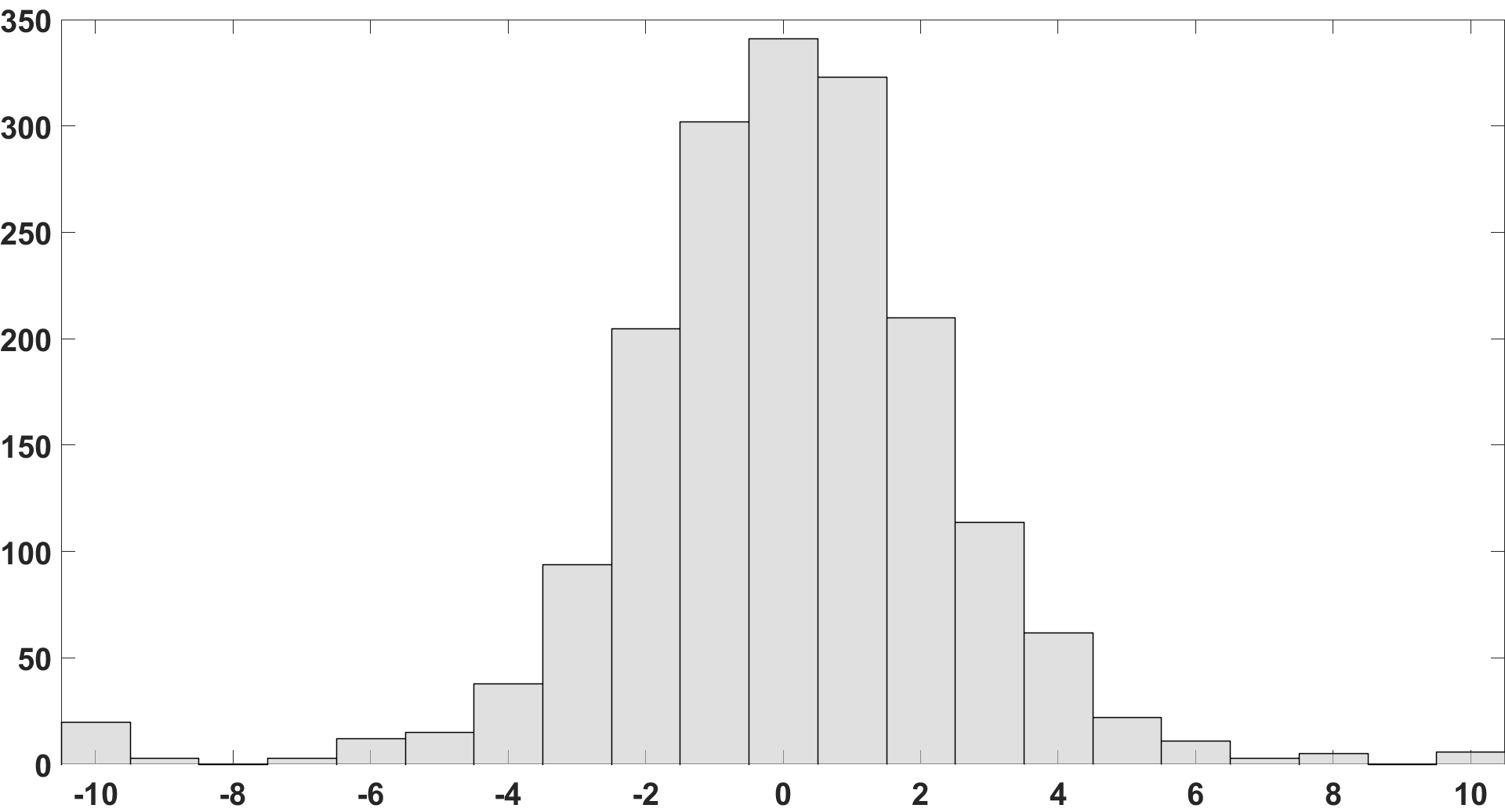}}

\put(4,45){\tiny MCYT}
\put(4,43){\tiny $M_\vecx$}
\put(4,21){\tiny VeriFinger}
\put(4,19){\tiny $M_\vecx$}

\put(49,45){\tiny MCYT}
\put(49,43){\tiny $M_{\vecx\qy}$}
\put(49,21){\tiny VeriFinger}
\put(49,19){\tiny $M_{\vecx\qy}$}

\end{picture}
\caption{\it Histograms of the optimal rotation angle~$\qf_0$ (degrees).}
\label{fig:rothist}
\end{center}
\end{figure}

\begin{figure}
\begin{center}
\setlength{\unitlength}{1mm}
\begin{picture}(90,100)(0,0)

\put(0,52){\includegraphics[width=88mm]{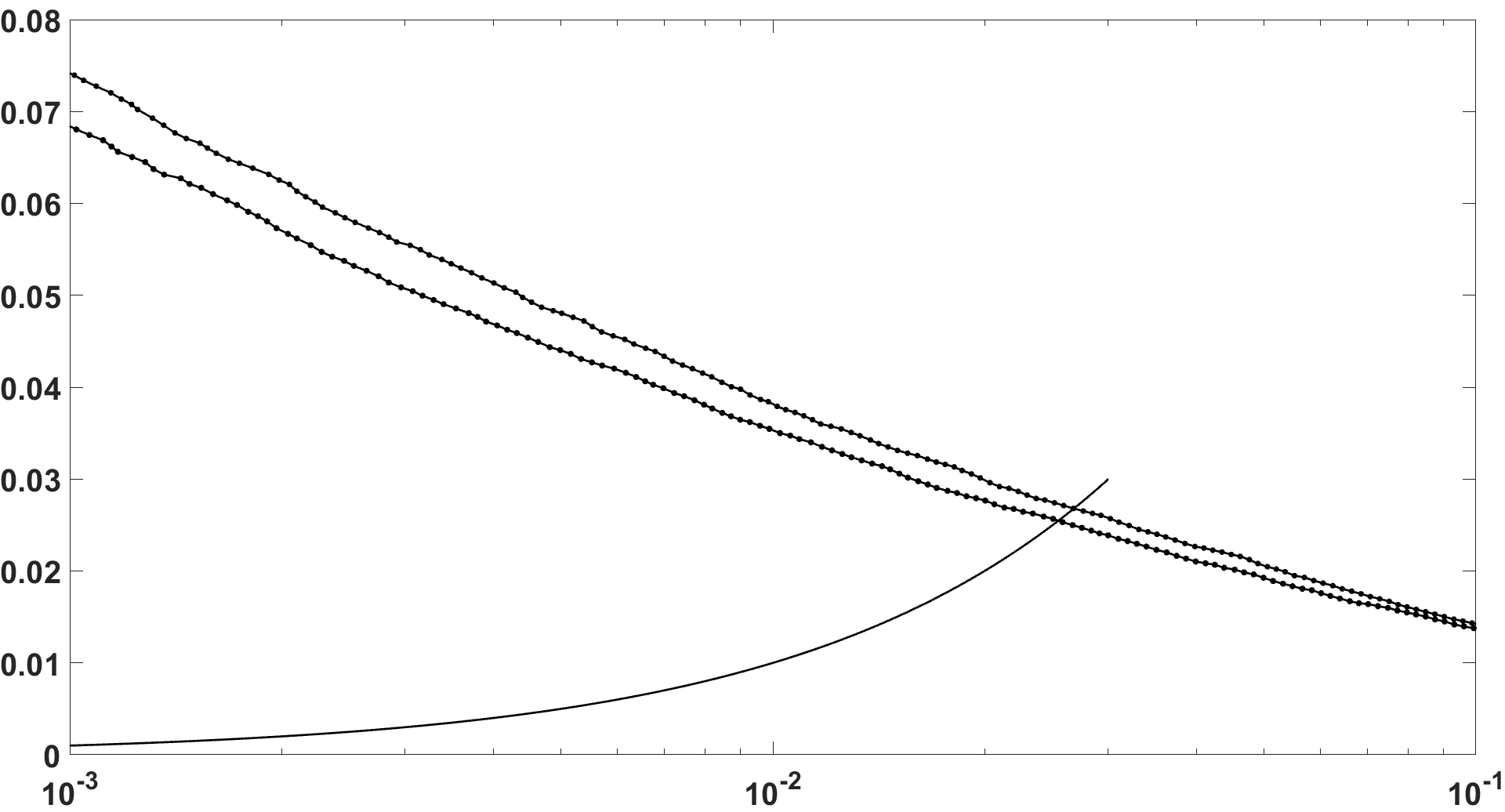}}
\put(35,95){MCYT database, $M_{\vecx\qy}$}
\put(57,64){\tiny FAR=FRR}
\put(77,52.5){\small FAR}
\put(5,96){\small FRR}
\put(35,81){\small without rotation}
\put(14,78){\small with rotation}
\put(14,74){\small $|\qf|\leq 3^\circ$}

\put(0,01){\includegraphics[width=88mm]{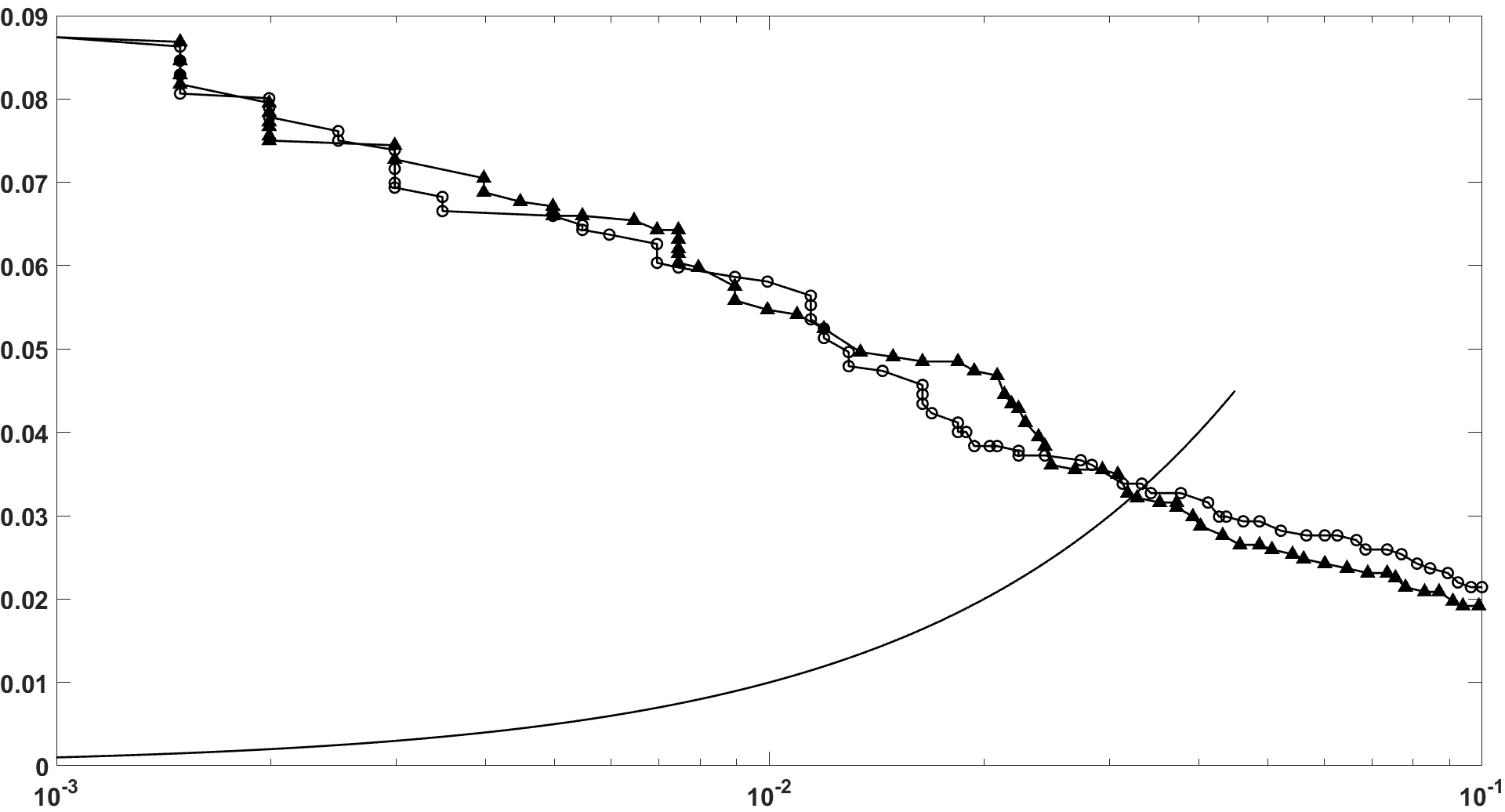}}
\put(35,44){VeriFinger database, $M_{\vecx\qy}$}
\put(57,11){\tiny FAR=FRR}
\put(77,1){\small FAR}
\put(4,43){\small FRR}
\put(52,28){\small with rotation, $|\qf|\leq 4.5^\circ$}
\put(35,21){\small without rotation}

\end{picture}
\caption{\it ROC curves with and without rotation of the verification image.}
\label{fig:rotROC}
\end{center}
\end{figure}

\section{Computational efficiency}
\label{sec:efficiency}

{\em In this analysis we do not use the potential speedup that can be gained by
skipping the $\qf$-search.}

Speed is important predominantly in the verification phase.
From a freshly captured image the spectral function has to be computed on a number of
grid points which we denote as $N_{\rm gr}$.
The spectral function has to be computed not once but several times,
because $N_\qf$ different image orientations have to be tried.
Fortunately this does not multiply the total effort\footnote{
Here we look only at the computation of the spectral function and the score;
not at the cost of $N_\qf$ Secure Sketch reconstruction attempts. 
}
by a factor $N_\qf$,
as the spectral function has a simple transform under rotation.
(This holds for Xu et al as well as our $L$ and $M$ functions.) 

Let $Z$ be the number of minutiae.
Let us denote the cost of computing one summation term of the spectral function in one grid point as 
$T_{\rm s}$, and the cost of applying a rotation transform in one grid point as
$T_{\rm rot}$.
The cost of computing the score can be written as $c\cdot N_{\rm gr}$
where $c$ is some small constant.
The superscript `$G$' will refer to Xu et al's spectral function;
the superscript `$M$' to our function~$M$.
The total cost for the verification phase (not counting the secure sketch) is
\bea
	&& \!\!\!\!\!\!\!
	\mbox{Xu et al:}\;\;\;\;  
	N_{\rm gr}^G Z T_{\rm s}^G + (N_\qf-1)N_{\rm gr}^G T_{\rm rot}^G +N_\qf cN_{\rm gr}^G 
	\nn\\ &&
	\!\!\!\!\!\!\!\!\!\!\!\!\!
	\mbox{pair-based:}\, 
	N_{\rm gr}^M{Z\choose 2} T_{\rm s}^M + (N_\qf\!-\!1)N_{\rm gr}^M T_{\rm rot}^M +N_\qf cN_{\rm gr}^M.
	\nn
\eea
We have $T_{\rm s}^G\approx T_{\rm s}^M$, $T_{\rm rot}^G\approx T_{\rm rot}^M$, $T_{\rm rot}< T_{\rm s}$.
The main difference between the two approaches lies in the first term:
$N_{\rm gr}^G Z$ versus $N_{\rm gr}^M {Z\choose 2}$,
i.e. $N_{\rm gr}^G$ versus $\fr12 N_{\rm gr}^M (Z-1)$.
Xu et al report a $128\times 256$ grid, yielding 
$N_{\rm gr}^G=32768$.
In contrast, our $M_{\vecx\qy}$-function is evaluated on a grid of size
$N_{\rm gr}^M\leq 16\cdot 25=400$.
Given that typically $Z\approx 35$,
we have $\fr12 N_{\rm gr}^M (Z-1) \approx 6800$.
Hence our verification is faster than \cite{XVBKAG2009,XV2009,XV2009CISP}.

Note that \cite{XuVeldhuis2010} introduces a reduced template size by
applying Principal Component Analysis or a Discrete Fourier Transform
to select informative features. This selection reduces the template size by roughly a factor~10.
However, these methods still require computation of the spectral function on many grid points.

\section{Discussion}
\label{sec:discussion}

Achieving translation invariance by looking at
{\em minutia pairs} seems to be advantageous compared to
taking the absolute value of a Fourier transform.
The minutia-pair approach is able to extract information from a fingerprint 
using fewer grid points.
We conjecture that this is due to the fact that our spectral functions
retain phase information instead of discarding it.
Of the four functions that we studied, the $M_{\vecx\qy}$ performs best.
Fusion of the matching scores from $M_\vecx$ and $M_{\vecx\qy}$
leads to an EER comparable to Xu et al.

Due to the reduction of the number of grid points
our method is faster than the verification described by Xu et al., in spite of  
the increased number of summation terms.
As an unexpected bonus, it turns out that we can omit the search for an optimal rotation angle;
this gives an additional speed improvement.

As topics for future work we mention (i) further speedup
by discarding grid points that have a bad signal-to-noise ratio;
(ii) applying Principal Component Analysis and similar techniques to improve the EER;
(iii) constructing a HDS based on $M_\vecx$ and $M_{\vecx\qy}$.

\bibliographystyle{plain}
\bibliography{minutiae}

\end{document}